\documentclass{article}

\usepackage{times}
\usepackage[pdftex]{graphicx}
\usepackage{amssymb,amsfonts,amsmath,amsthm}
\usepackage{float}
\usepackage{xcolor}
\newtheorem{theorem}{Theorem}
\newtheorem{proposition}{Proposition}
\newtheorem{lemma}{Lemma}
\newtheorem{ass}{Assumption}
\newtheorem{definition}{Definition}
\newtheorem{remark}{Remark}

\newcommand\EE {\mathbb E}
\newcommand\FF {\mathbb F}

\newcommand\RR {\mathbb R}
\newcommand\PP {\mathbb P}

\newcommand\ZZ {\mathbb Z}

\def\bone{\mathbf{1}}

\def\qed{\hskip6pt\vrule height6pt width5pt depth1pt}

\def\qed{\hskip 6pt\vrule height6pt width5pt depth1pt}

\newcommand{\ed}{\end{document}}
\newcommand{\be}{\begin{equation}}
\newcommand{\ee}{\end{equation}}
\newcommand{\bq}{\begin{eqnarray}}
\newcommand{\eq}{\end{eqnarray}}

\vspace{4in}

\newcommand{\supp}{\text{supp}}

\newcommand{\Fr}{\mathbf{F}}

\begin{document}

\title{Control-stopping Games for Market Microstructure and Beyond.\footnote{Partial support from the NSF grant DMS-1411824 is acknowledged by both authors.}}
\author{Roman~Gayduk and Sergey~Nadtochiy\footnote{Address the correspondence to: Mathematics Department, University of Michigan, 530 Church St, Ann Arbor, MI 48104; sergeyn@umich.edu.}}
\date{Current version: Nov 1, 2018
\vskip 4pt
Original version: Aug 1, 2017
}
\maketitle

\begin{abstract}
In this paper, we present a family of a control-stopping games which arise naturally in equilibrium-based models of market microstructure, as well as in other models with strategic buyers and sellers. A distinctive feature of this family of games is the fact that the agents do not have any exogenously given fundamental value for the asset, and they deduce the value of their position from the bid and ask prices posted by other agents (i.e. they are pure speculators). As a result, in such a game, the reward function of each agent, at the time of stopping, depends directly on the controls of other players.
The equilibrium problem leads naturally to a system of coupled control-stopping problems (or, equivalently, Reflected Backward Stochastic Differential Equations (RBSDEs)), in which the individual reward functions (or, reflecting barriers) depend on the value functions (or, solution components) of other agents.
The resulting system, in general, presents multiple mathematical challenges due to the non-standard form of coupling (or, reflection). In the present case, this system is also complicated by the fact that the continuous controls of the agents, describing their posted bid and ask prices, are constrained to take values in a discrete grid. The latter feature reflects the presence of a positive tick size in the market, and it creates additional discontinuities in the agents' reward functions (or, reflecting barriers). Herein, we prove the existence of a solution to the associated system in a special Markovian framework, provide numerical examples, and discuss the potential applications.
\end{abstract}

\section{Buyer-seller game}

In this paper, we consider a non-zero-sum control-stopping game between two players, which can be stated in the form
\begin{equation}\label{eq.game.main.def}
\left\{
\begin{aligned}
{(\theta^a,\tau^a)\in\text{argmax}_{\theta,\,\tau} \,\EE\left( U^a_{\tau}(\theta,\theta^b)\,\bone_{\{\tau\leq \tau^b\}} 
+ L^a_{\tau^b}(\theta,\theta^b)\,\bone_{\{\tau > \tau^b\}} \right)}\\
{(\theta^b,\tau^b)\in\text{argmin}_{\theta,\,\tau} \,\EE\left( U^b_{\tau}(\theta^a,\theta)\,\bone_{\{\tau\leq \tau^a\}} 
+ L^b_{\tau^a}(\theta^a,\theta)\,\bone_{\{\tau > \tau^a\}} \right)}
\end{aligned}
\right.
\end{equation}
In the above, $\theta$ is a stochastic process, referred to as continuous control; $\tau$ represents a stopping time (its exact meaning is discussed at the end of this section); $U^a$, $U^b$, $L^a$, and $L^b$, are exogenously given non-anticipative random functions, mapping the paths of $(\theta^a,\theta^b)$ into the so-called reward paths. The interpretation of (\ref{eq.game.main.def}) is clear: the first player, referred to as the seller, chooses optimal control $(\theta^a,\tau^a)$, and the second player, the buyer, chooses optimal $(\theta^b,\tau^b)$. The payoff of each player depends on her individual continuous control, on the continuous control of the other player, and on who stops first.
The general class of games defined by (\ref{eq.game.main.def}) can be viewed as an extension of the well known Dynkin games (\cite{Dyn1}), where the extension is due to the fact that (i) the players can choose continuous controls, in addition to stopping times, and (ii) the game is not necessarily zero-sum. The two-player non-zero-sum control-stopping games have been considered in \cite{KaratzasLi}, \cite{KaratzasSuderth}. Other extensions of Dynkin games can be found, e.g., in \cite{BensoussanFriedman}, \cite{Dyn3}, \cite{HamadenLapeltier}, \cite{HamadeneZhang}, \cite{AngelisFerrariMoriarty}, \cite{ZhouZhou}, and the references therein.
However, the specific type of games considered herein has not been analyzed before, with the only exception of our earlier work, \cite{GaydukNadtochiy2}. Namely, we assume that $\theta=(p,q)$, with real-valued processes $p$, $q$, and
\begin{equation}\label{eq.UL.def}
\begin{aligned}
U^{a}_t(p^a,p^b,q^b) = \int_0^t \exp\left(-\int_0^u c^a_s(p^a_s,p^b_s)ds\right) g^a_u(p^a_u,p^b_u)\,du 
\,+\,\exp\left(-\int_0^t c^a_s(p^a_s,p^b_s)ds\right)\, q^b_t,\\
L^{a}_t(p^a,q^a,p^b) = \int_0^t \exp\left(-\int_0^u c^a_s(p^a_s,p^b_s)ds\right) g^a_u(p^a_u,p^b_u)\,du 
\,+\,\exp\left(-\int_0^t c^a_s(p^a_s,p^b_s)ds\right)\, q^a_t,\\
U^{b}_t(p^a,q^a,p^b) = \int_0^t \exp\left(-\int_0^u c^b_s(p^a_s,p^b_s)ds\right) g^b_u(p^a_u,p^b_u)\,du 
\,+\,\exp\left(-\int_0^t c^b_s(p^a_s,p^b_s)ds\right)\, q^a_t,\\
L^{b}_t(p^a,p^b,q^b) = \int_0^t \exp\left(-\int_0^u c^b_s(p^a_s,p^b_s)ds\right) g^b_u(p^a_u,p^b_u)\,du 
\,+\,\exp\left(-\int_0^t c^b_s(p^a_s,p^b_s)ds\right)\, q^b_t.
\end{aligned}
\end{equation}
This choice has a clear economic interpretation, as a game between buyers and sellers, which is the main motivation for our study of such problems (cf. \cite{GaydukNadtochiy1}, \cite{GaydukNadtochiy2}). Consider a strategic buyer and a strategic seller (or two homogeneous groups of buyers and sellers), trying to buy and sell, respectively, one unit of an asset. Assume that the strategic agents know each other's characteristics (e.g. because they have played this game many times, or because they can ``research" each other), but there exist other potential buyers and sellers, who are not fully strategic, and whose characteristics are not completely known to the strategic players. An example of such setting appears naturally in the context of market microstructure, where professional high-frequency traders use common predictive factors and, hence, may, to some extent, predict the behavior of each other, while the longer term investors (or, simply, less sophisticated ones) may use very different types of strategies and, hence, are more difficult to predict. A similar game may describe the behavior of a department store and the ``bargain hunters", representing the strategic sellers and buyers, respectively. The department store may know fairly well the preferences (and, hence, the behavior) of the bargain hunters by conducting market research or through past experience. However, there may also exist other potential shoppers, whose behavior is more difficult to predict: e.g. it may be motivated more by the immediate personal needs than by the prices and discounts. Regardless of the specific economic interpretation, the behavior of the strategic agents can be described in the following way. The strategic agents model the arrival time of the first non-strategic agent via an exponential random variable, independent of the past information they observed. The non-strategic agents, despite being non-strategic, do have price preferences, hence, the arrival rate of their orders, $c^a_t$ or $c^b_t$ (we allow the strategic agents to have different beliefs about the arrival rate), depends on the prices posted by the strategic buyers and sellers, $p^b_t$ and $p^a_t$, respectively. The latter also affect the probability that the arriving non-strategic agent is a buyer or a seller and, hence, will trade with the strategic seller or the strategic buyer, respectively. The game ends whenever the first trade occurs. If the arriving non-strategic agent is a seller, then, the strategic seller (who, in this case, is left with positive inventory at the end of the game), marks her inventory to the ``new fair price", which is typically lower than her posted price, $p^a$. Analogous rule applies if the non-strategic agent is a buyer. The specific choice of such ``marking to market" of the remaining inventory is discussed in the subsequent section. The resulting profits and losses are captured by the terms $g^a$ and $g^b$, in (\ref{eq.UL.def}). However, the game may also end if the strategic buyer and seller decide to trade with each other. The latter may occur if one of them becomes sufficiently pessimistic about the arrival of a non-strategic agent of appropriate type. In this case, they will trade at the so-called ``internal" prices, $q^a$ and $q^b$, which are not posted (and, hence, not known to the non-strategic agents) but can be deduced by each strategic agent in equilibrium. A strategic agent who initiates the trade accepts the price offered by the other agent: e.g. if the strategic seller initiates a trade with the strategic buyer, the transaction will occur at the price $q^b$. This is captured by the last terms in the right hand side of (\ref{eq.UL.def}). A more detailed mathematical description of the market microstructure model leading to  (\ref{eq.game.main.def})--(\ref{eq.UL.def}) is presented in \cite{GaydukNadtochiy2}.

Let us now explain why the games in the form (\ref{eq.game.main.def})--(\ref{eq.UL.def}) are challenging to analyze. First of all, the fact that such games are not zero-sum implies that the description of equilibrium through a doubly-reflected BSDE, established, e.g., in \cite{Dyn3}, \cite{HamadenLapeltier}, is not applicable in the present case. In the absence of zero-sum property, the existence of equilibrium is typically established via a fixed-point theorem -- either using its abstract version (cf. \cite{KaratzasLi}, \cite{KaratzasSuderth}, \cite{HamadeneZhang}, \cite{ZhouZhou}), or in the context of Partial Differential Equations (cf. \cite{BensoussanFriedman}).\footnote{An exception is \cite{AngelisFerrariMoriarty}, where an equilibrium is constructed explicitly. However, the latter results rely on the specific structure of the problem, and, in particular, do not allow for continuous controls.} In order to apply a fixed-point theorem, one has to either rely on certain monotonicity properties of the objective, which are not present in our case, or (a) to find a compact set of individual controls, which is sufficiently large to include any maximizer of the objective function, and (b) to establish the continuity of the objective with respect to all controls.
However, the form of the last terms in the right hand side of (\ref{eq.UL.def}), representing the ``reward" at stopping, makes it difficult to implement the program (a)--(b). Notice that these terms introduce expressions of the form ``$q_{\tau}$" into the objective. It turns out that it is very difficult to find a topology on the space of controls $(q,\tau)$, such that the set of all admissible controls is compact and the mapping $(q,\tau)\mapsto q_{\tau}$ is continuous. Note that this issue does not arise in many cases considered in the existing literature, because the instantaneous reward process, typically, depends on the continuous controls only through a controlled state process (cf. \cite{KaratzasLi}, \cite{KaratzasSuderth}). The latter type of dependence is more convoluted than the explicit functional dependence on the continuous controls, given by the last terms in the right hand side of (\ref{eq.UL.def}), but it does have a ``smoothing" effect and, in many cases, allows one to avoid the aforementioned issue of the lack of compactness or continuity.
Our earlier work, \cite{GaydukNadtochiy2}, shows how to construct equilibria in a version of (\ref{eq.game.main.def})--(\ref{eq.UL.def}), using a system of Reflected Backward Stochastic Differential Equations (RBSDEs).

Herein, we extend our discussion of the control-stopping games in the form (\ref{eq.game.main.def})--(\ref{eq.UL.def}) and consider the case where the continuous controls $(p^a,q^a,p^b,q^b)$ are only allowed to take values in a \emph{discrete equidistant grid}, which, without loss of generality, is assumed to be $\ZZ$.
This restriction reflects the presence of a strictly positive \emph{tick size} in financial markets: i.e., a transaction can only occur at a price that is a multiple of the tick size.
It turns out that, in the presence of a positive tick size, the equilibrium dynamics become qualitatively different, and the construction of equilibrium requires novel mathematical methods, as compared to the related model of \cite{GaydukNadtochiy2}. These two features constitute the main contribution of the present paper.
In particular, the equilibrium constructed herein allows for the continuation values of the two players to \emph{cross}: i.e., a strategic seller may expect to receive for his asset less than the strategic buyer expects to pay. In such a case, without the friction caused by a positive tick size, one would naturally expect that the two strategic agents would trade with each other at any price level lying between the two continuation values, and the game would end. The latter is indeed the case in \cite{GaydukNadtochiy2}. However, in the present model, the agents may not be able to trade with each other because there may not exist an admissible price level (i.e. a multiple of tick size) lying between their continuation values. In fact, such a crossing of continuation values occurs necessarily if the two strategic agents have the same beliefs about the future order flow, in which case the game ends instantaneously in the model of \cite{GaydukNadtochiy2}, whereas it typically continues in the present model. The size of such crossing (i.e., by how much the continuation value of the buyer exceeds the continuation value of the seller) measures the \emph{inefficiency} created by the positive tick size. Namely, if the two players were offered to trade with each other in a ``shadow market", without a tick size, the maximum fee they would be willing to pay for such an opportunity is precisely equal to the size of the crossing. The results of the present paper allow us to compute the crossing size and to study its dependence on the tick size. This analysis is presented in Section \ref{se:example}.

On the mathematical side, our main contribution is in showing how to prove the existence of equilibrium in a game in which the individual objectives of the players depend on a discontinuous (but piece-wise continuous) function of the controlled state process, such as the ``floor" and ``ceiling" functions.
Such dependencies appear naturally in various games. For example, a series of attacks on a computer server does not cause any loss to its owner until the server's security is breached, in which case the loss jumps instantaneously. Similarly, changes in credit quality of a borrower do not cause any loss to his creditor until the credit quality drops to a critical value, at which the borrower defaults on his loan and the loss jumps. In the present case, it turns out that there exists an equivalent formulation of the equilibrium problem in which the internal prices $(q^a,q^b)$ of strategic agents are given by the ``floor" and ``ceiling" functions of the respective continuation values, introducing a potential discontinuity in the objectives. It is well known that the existence of a fixed point for a discontinuous mapping is quite rare, and proving it is a very challenging task. However, if the controlled state process is sufficiently ``diffuse",\footnote{The precise interpretation of this term depends on the equilibrium problem at hand. Herein, we show that the continuation values of strategic agents are given by the functions of a Brownian motion, with the first derivatives being bounded away from zero.} the discontinuities are smoothed out by the optimization. Thus, the method we propose and implement herein is to restrict the controls of the players in such a way that (i) the controlled state process remains sufficiently diffuse and (ii) there is no loss of optimality due to such restriction, and to show the existence of a fixed point in the restricted set. A more detailed discussion of our implementation of this general method is given after the statement of Theorem \ref{thm:main}.

The paper is organized as follows. In Section \ref{sec.2agents}, we describe the control-stopping game in a Markovian framework and obtain an alternative representation of its equilibria via an auxiliary system of coupled control-stopping problems, in which the immediate reward process of one problem depends directly on the value function of the other one (cf. (\ref{eq.2agentjointproblem})). The main result of the paper, Theorem \ref{thm:main}, is also stated in Section \ref{sec.2agents}, and it proves the existence of a solution to the aforementioned auxiliary system (and, hence, shows the existence of equilibrium in the proposed game). Most of the remainder of the paper is devoted to the proof of the main result. 
In Section \ref{sec.individualvalfuncprops}, we investigate the properties of the value function of an individual optimal stopping problem faced by an agent, given the continuous controls of both agents and the other agent's value function. Section \ref{sec.individualvalfuncprops} describes the main novel mathematical ideas developed in this paper (see the discussion following Theorem \ref{thm:main}). In Section \ref{sec.responsepricesandequilibrium}, we, first, show that the candidate optimal continuous controls in feedback form are, indeed, optimal, and that they possess certain continuity properties. We, then, use these properties and the results of Section \ref{sec.individualvalfuncprops} to prove the existence of a solution to the system of coupled control-stopping problems derived in Section \ref{sec.2agents}, via a fixed-point theorem.
A numerical example, illustrating the properties of proposed equilibrium models, and its applications, are discussed in Section \ref{se:example}.


\section{Problem formulation and main results}\label{sec.2agents}


The continuous controls of each agent take values in the set of integers $\mathbb{Z}$, which corresponds to measuring prices in the multiples of tick size, in the context of the market microstructure model.
The information of each agent is generated by (the same) random factor
$$
X_t = X_0 + \sigma B_t,
$$
where $X_0\in\RR$ and $B$ is a one-dimensional Brownian motion.
The rest of the specifications are given in order to determine a specific form of $c^{a/b}_t$ and $g^{a/b}_t$, and are motivated by the model of market microstructure, discussed, e.g., in \cite{GaydukNadtochiy2}. However, most of the results described herein hold for more general $c^{a/b}_t$ and $g^{a/b}_t$, as long as they are given by deterministic functions of $(p^a_t,p^b_t,X_t)$ and the appropriate assumptions, stated further in the paper, are satisfied.

In the proposed specification, $X$ represents the current \emph{estimate} of the ``reservation price" of the non-strategic investors. Each agent believes that the non-strategic investors arrive to the market according to a Poisson process $N$, with intensity $\lambda>0$, independent of $B$.
At any arrival time $T_i$ of $N$, the value of the \emph{actual} reservation price, $p^0_{T_i}$, of a non-strategic investor arriving at this time, is given by 
$$
p^0_{T_i}=X_{T_i}+\xi_i,
$$ 
where $\{\xi_i\}$ are i.i.d. random variables, independent of $B$ and $N$, with mean $0$ and c.d.f. $F$.
An arriving non-strategic investor submits a buy order (to the strategic seller) if and only if $p^0_{T_i}\geq p^a_{T_i}$. Similarly, she submits a sell order (to the strategic buyer) if and only if $p^0_{T_i}\leq p^b_{T_i}$.
After the order of a non-strategic seller is executed (all orders are of unit size), the remaining inventory of strategic seller (of unit size) is marked to market at the price $\lfloor X_{T_i-}+\alpha\xi_{i} \rfloor$. Similarly, if the non-strategic agent is a buyer, the remaining inventory of strategic buyer is marked to market at the price $\lceil X_{T_i-}+\alpha\xi_{i} \rceil$. We denote by $\lfloor\cdot\rfloor$ and $\lceil\cdot\rceil$ the standard ``floor" and ``ceiling" functions. The parameter $\alpha\in(0,1)$ measures the (linear) impact of a single trade on the estimated fundamental price. The choice of this specific marking rule is justified if the bid-ask spread is equal to a single tick: in this case, $X_{T_i-}+\alpha\xi_{i}$ represents the new estimate of the reservation price (after the order of the non-strategic agent is executed), and the bid and ask prices are the closest integers to this number from below and from above.
The above specification leads to the following expressions for $c^{a/b}$ and $g^{a/b}$:
\begin{equation}\label{def.cpapbx}
c^{a}_t(p^a_t,p^b_t) = c^{b}_t(p^a_t,p^b_t) = c(p^a_t,p^b_t,X_t),
\quad c(p^a,p^b,x):=\lambda\left( \left(1-F\left(p^a-x\right)\right)+F\left(p^b-x\right)\right),
\end{equation}
\begin{equation}\label{def.gapapbx}
g^{a}_t(p^a_t,p^b_t) = g^a(p^a_t,p^b_t,X_t),
\quad g^a(p^a,p^b,x):=\lambda\left(
p^a\left( 1-F\left(p^a-x\right)\right)+\mathcal{F}^b(p^b,x)
\right),
\end{equation}
\begin{equation}\label{def.gbpapbx}
g^{b}_t(p^a_t,p^b_t) = g^b(p^a_t,p^b_t,X_t),
\quad g^b(p^a,p^b,x):=\lambda\left(
p^bF\left(p^b-x\right)+\mathcal{F}^a(p^a,x)
\right),
\end{equation}
\begin{equation}\label{def.Fbpbx}
\mathcal{F}^b(p^b,x):=\int_{-\infty}^{p^b-x} \lfloor x+\alpha y\rfloor\text{d}F(y),
\quad \mathcal{F}^a(p^a,x):=\int^{\infty}_{p^a-x} \lceil x+\alpha y\rceil\text{d}F(y).
\end{equation}
Any strategic agent can also stop the game at any point in time by trading with the other strategic agent at the price acceptable to the latter: $q^a$, if the buyer initiates the trade, and $q^b$ if the seller does. This creates the optimal stopping problems for strategic agents, with the instantaneous reward processes being equal to (discounted) $q^a$ and $q^b$ (see (\ref{eq.game.main.def}) and (\ref{eq.UL.def})).
Such functional dependence of the instantaneous reward process on the continuous control of the other agent creates an additional challenge in the definition of stopping strategies. Assume, for simplicity, that $c^{a/b}\equiv0$ and notice that, if $\tau^a>\tau^b$, then, the payoff of the seller is $q^a_{\tau^b}$. If $\tau^b$ is fixed as a stopping time (and known to the seller), then the seller can increase the value of $q^a$ around $\tau^b$, thus, making her payoff arbitrarily large. Of course, in practice, the buyer would not accept an arbitrarily high price. Hence, to make our game more realistic, we assume that each stopping strategy depends on the continuous control of the other agent, so that changing $q^a$ would cause changes in $\tau^b$.

The resulting control-stopping game is given by
\begin{equation}\label{eq.game.main.def.2}
\left\{
\begin{aligned}
(p^a,q^a,\tau^a)\in\text{argmax}_{p\in\mathcal{A}^a,\,q\in\mathcal{A},\,\tau\in\mathcal{T}^a}& \,\EE\left( U^a_{\tau(q^b)}(p(X),p^b(X),q^b(X))\,\bone_{\{\tau(q^b)\leq \tau^b(q)\}} \right.\\
&\left.\quad\quad\quad\quad + L^a_{\tau^b(q)}(p(X),q(X),p^b(X))\,\bone_{\{\tau(q^b) > \tau^b(q)\}} \right),\\
(p^b,q^b,\tau^b)\in\text{argmin}_{p\in\mathcal{A}^b,\,q\in\mathcal{A},\,\tau\in\mathcal{T}^b}&\,\EE\left( U^b_{\tau(q^a)}(p^a(X),q^a(X),p(X))\,\bone_{\{\tau(q^a)\leq \tau^a(q)\}} \right.\\
&\left.\quad\quad\quad\quad + L^b_{\tau^a(q)}(p^a(X),p(X),q(X))\,\bone_{\{\tau(q^a) > \tau^a(q)\}} \right),
\end{aligned}
\right.
\end{equation}
where $U^{a/b}$ and $L^{a/b}$ are defined via (\ref{eq.UL.def}) and (\ref{def.cpapbx})--(\ref{def.Fbpbx}); $\mathcal{A}$ is the space of all measurable functions from $\RR$ to $\mathbb{Z}$; 
\begin{equation}\label{eq.Aa.Ab.def}
\mathcal{A}^a = \{p\in\mathcal{A}\,:\,1-F(p(x)-x) \geq \frac{c_l}{2\lambda},\,\,\forall\,x\in\RR\},
\quad \mathcal{A}^b = \{p\in\mathcal{A}\,:\,F(p(x)-x) \geq \frac{c_l}{2\lambda},\,\,\forall\,x\in\RR\},
\end{equation}
with some constant $c_l>0$; and $\mathcal{T}^a$ and $\mathcal{T}^b$ consist of all mappings from $\mathcal{A}$ into the space of $\FF^X$-stopping times, s.t.
$$
\tau\in\mathcal{T}^a\Rightarrow\exists\, \FF^X\text{-adapted process }v,\text{ s.t. } \forall\,q\in\mathcal{A},\,\,\tau(q) = \inf\{t\geq0\,:\,q(X_t) \geq v_t\},
$$
$$
\tau\in\mathcal{T}^b\Rightarrow\exists\, \FF^X\text{-adapted process }v,\text{ s.t. } \forall\,q\in\mathcal{A},\,\,\tau(q) = \inf\{t\geq0\,:\,q(X_t) \leq v_t\}.
$$
The above definition of stopping strategies reflects the conclusion of the paragraph preceding (\ref{eq.game.main.def.2}) and implies that an agent's stopping rule is of ``threshold" type: she stops whenever the continuous control of the other agent (i.e. the quoted price) becomes sufficiently attractive. One can uniquely define any $\tau$ in $\mathcal{T}^a$ or in $\mathcal{T}^b$ by specifying the associated threshold $v$. We will, therefore, sometimes, use the notation $\tau^v$. Clearly, for any fixed $q\in\mathcal{A}$, any $\FF^X$-stopping time can be represented as $\tau^v(q)$, with an appropriately chosen $v$.
The sets of admissible posted prices, $\mathcal{A}^{a/b}$, are chosen to ensure that the rate of the order flow from non-strategic agents, $c(p^a,p^b,x)$ (given in (\ref{def.cpapbx})), is bounded away from zero.\footnote{This is a technical condition needed to ensure that the game does not last too long, and, as a result, the agents' value functions do not deviate too far from $X$ (see, e.g., (\ref{eq.cgeqcl}) and the proof of Lemma \ref{lemma.Cfromx}).}

To ensure that the optimal continuous control of each agent is well-behaved, we make the following assumption (compare to Assumptions 2--5 in \cite{GaydukNadtochiy2}).
\begin{ass}\label{ass.Foverfmono}
The distribution of $\xi$ has density $f$, bounded from above by a constant $C_f>0$, whose support is convex and contained in $[-C_0,C_0]$, for some constant $C_0>0$. Furthermore, in the interior of the support, $f$ is continuous, $(1-F)/f$ is non-increasing, and $F/f$ is non-decreasing.
\end{ass}
\begin{remark}
Assumption \ref{ass.Foverfmono} can be relaxed so that $(1-F)/f$ is only required to be non-increasing in the interior of $\text{supp}(f)\cap\RR_+$, and $F/f$ is non-decreasing in the interior of $\text{supp}(f)\cap\RR_-$. This would require a minor change in the modeling assumptions: namely, an external buy order can only arrive at a positive jump of $p^0$, and an external sell can only arrive at a negative jump of $p^0$ (in addition to $p^0\geq p^a$ and $p^0\leq p^b$, respectively). The latter setting is used in \cite{GaydukNadtochiy2}, but we choose to avoid it here, in order to streamline the notation. In addition, as discussed in the remark following Assumption 5 in \cite{GaydukNadtochiy2}, a sufficient condition for Assumption \ref{ass.Foverfmono} can be stated in terms of the log-concavity of $f$.
\end{remark}

As discussed in the introduction, games in the form (\ref{eq.game.main.def})--(\ref{eq.UL.def}) cannot be analyzed directly using standard methods.
Thus, our first goal is to find a more convenient system of equations describing a sub-class of equilibria of (\ref{eq.game.main.def.2}). To this end, we recall the notion of a value function, which represents the supremum or infimum (whichever one is appropriate) of the objective value of an agent over all admissible controls. It is easy to deduce that this value depends only on the initial condition $X_0$. Let us denote by $\bar{V}^a(x)$ and $\bar{V}^b(x)$ the value functions of the seller and the buyer, respectively, given $X_0=x$. It is natural to expect that, in equilibrium, $q^a(x)\geq \bar{V}^a(x)$ and $q^b(x)\leq \bar{V}^b(x)$: the optimal prices at which the agents are willing to trade immediately should not be worse than the execution price they expect to receive if they act optimally. Since the prices take integer values, we conclude: $q^a(x)\geq \lceil\bar{V}^a(x)\rceil$ and $q^b(x)\leq \lfloor\bar{V}^b(x)\rfloor$. Then, it is suboptimal for any agent to stop unless either $\bar{V}^b(X_t)\geq \lceil\bar{V}^a(X_t)\rceil$ or $\bar{V}^a(X_t)\leq \lfloor\bar{V}^b(X_t)\rfloor$. Notice that each of the latter two inequalities describes the same event: the interval $[\bar{V}^a(X_t),\bar{V}^b(X_t)]$ contains at least one integer.
Then, it is optimal for each strategic agent to trade with the other one at the price given by any integer in this interval, and it is suboptimal for at least one of them to trade at any other price. If such integer is unique, we obtain the unique threshold price such that the agents stop (simultaneously) when their value functions reach this threshold. This heuristic discussion motivates the search for an equilibrium of (\ref{eq.game.main.def.2}) via the auxiliary system of coupled control problems: find $p^a\in\mathcal{A}^a$, $p^b\in\mathcal{A}^b$, and measurable $\bar{V}^a$ and $\bar{V}^b$, such that

\begin{equation}\label{eq.2agentjointproblem}
\begin{cases}
\bar{V}^a(x)=\sup_{p\in\mathcal{A}^a,\,\tau} J^a\left(x,\tau,p,p^{b},\bar{V}^b\right)= \sup_\tau J^a\left(x,\tau,p^{a},p^{b},\bar{V}^b\right),\\
\bar{V}^b(x)=\inf_{p\in\mathcal{A}^b,\,\tau} J^b\left(x,\tau,p^{a},p, \bar{V}^a\right)= \inf_\tau J^b\left(x,\tau,p^{a},p^{b}, \bar{V}^a\right),
\end{cases}
\end{equation}
where $\tau$ changes over all $\FF^X$-stopping times,
and
\begin{equation}
\begin{split}\label{def.JagivenJb}
J^a\left(x,\tau,p^a,p^{b},v\right)=\EE^x\Big[\int_0^\tau \exp\left(-\int_0^t c(p^a(X_s),p^b(X_s),X_s)\text{d}s\right) g^a(p^a(X_t),p^b(X_t),X_t)\text{d}t \\ +\exp\left(-\int_0^\tau c(p^a(X_s),p^b(X_s),X_s)\text{d}s\right)\lfloor v(X_\tau)\rfloor\Big],
\end{split}
\end{equation}
\begin{equation}\label{def.JbgivenJa}
\begin{split}
J^b\left(x,\tau,p^a,p^{b},v\right)=\EE^x\Big[\int_0^\tau \exp\left(-\int_0^t c(p^a(X_s),p^b(X_s),X_s)\text{d}s\right) g^b(p^a(X_t),p^b(X_t),X_t)\text{d}t \\+\exp\left(-\int_0^\tau c(p^a(X_s),p^b(X_s),X_s)\text{d}s\right)\lceil v(X_\tau)\rceil\Big],
\end{split}
\end{equation}
with $\EE^x[\cdot]=\EE\left[\cdot\vert X_0=x\right]$.

\begin{proposition}\label{prop:equivEquil}
Assume that $(p^a,p^b,\bar{V}^a,\bar{V}^b)$, solve (\ref{eq.2agentjointproblem}), and that $(\bar{V}^a,\bar{V}^b)$ are continuous and increasing. Then, $(p^a, \lceil \bar{V}^a \rceil,\tau^{\bar{V}^a})$ and $(p^b,\lfloor \bar{V}^b\rfloor,\tau^{\bar{V}^b})$ form an equilibrium of (\ref{eq.game.main.def.2}).
\end{proposition}
\begin{proof}
The proof consists of elementary verification that the proposed strategies are optimal for the agents in (\ref{eq.game.main.def.2}).
Consider the strategic seller. From the fact that $q^b = \lfloor\bar{V}^b\rfloor\leq \bar{V}^a$, we easily deduce that the maximum objective value of the seller in (\ref{eq.game.main.def.2}) does not exceed $\bar{V}^a$, provided that the other agent uses the strategy given in the statement of the proposition. As $\bar{V}^b\leq \lceil\bar{V}^a\rceil = p^a$, it is easy to see that the agents using the prescribed strategies, in (\ref{eq.game.main.def.2}), both stop at the stopping time $\tau^*:=\inf\{t\geq0\,:\,\bar{V}^a(X_t) = \lfloor \bar{V}^b(X_t)\rfloor\}$, and that the resulting objective value of the seller in (\ref{eq.game.main.def.2}) coincides with the associated value of (\ref{def.JagivenJb}). The continuity and monotonicity of $\bar{V}^a$ and $\bar{V}^b$ imply: $\bar{V}^a(X_{\tau^*}) = \lfloor \bar{V}^b(X_{\tau^*})\rfloor$, whenever $\tau^*<\infty$. Then, the standard arguments in optimal stopping theory yield that the associated value of (\ref{def.JagivenJb}) is optimal and, hence, coincides with $\bar{V}^a$. Similar arguments apply to the strategic buyer.
\qed
\end{proof}

The above proposition shows that the problem of constructing an equilibrium of (\ref{eq.game.main.def.2}) reduces to solving (\ref{eq.2agentjointproblem}). The latter is a system of two Markovian control-stopping optimization problems, coupled through continuous controls and stopping barriers, with each barrier given by a discontinuous function (i.e. floor or ceiling) of the other agent's value function. Even in the present one-dimensional Brownian setting, the associated fixed-point problem lacks the desired continuity or monotonicity properties, rendering it intractable by general methods.
Nevertheless, the following theorem establishes the existence of a solution to (\ref{eq.2agentjointproblem}). In addition, its proof, presented in Sections \ref{sec.individualvalfuncprops} and \ref{sec.responsepricesandequilibrium}, and the discussions in Section \ref{se:example} shed more light on the structure of the solution.

\begin{theorem}\label{thm:main}
For any $C_f,\,C_0>0$, any c.d.f. $F$, satisfying Assumption \ref{ass.Foverfmono}, and any $c_l,\,\lambda>0$, $\alpha\in(0,1)$, there exists $\bar{\sigma}>0$, such that, for any $\sigma\geq \bar{\sigma}$, there exists a solution $(p^a,p^b,\bar{V}^a,\bar{V}^b)$ to (\ref{eq.2agentjointproblem}), with $(\bar{V}^a,\bar{V}^b)$ being continuous and increasing.\footnote{It is worth mentioning that the proof of Theorem \ref{thm:main}, in fact, shows the existence of a Markov perfect (or, sub-game perfect) equilibrium, in the sense that, if the agents re-evaluate their objectives at any intermediate time, the proposed strategies would still form an equilibrium.}
\end{theorem}

Let us outline the proof of Theorem \ref{thm:main}. First, we notice that the RBSDE approach of \cite{GaydukNadtochiy2} cannot be used in this case. Indeed, the domain of reflection for the associated system of RBSDEs is non-convex and discontinuous (see Section \ref{se:example}), and, to date, no existence results are available for such systems. Thus, we pursue a direct approach, aiming to show that the mapping $(v^a,v^b)\mapsto (J^a,J^b)\mapsto (\bar{V}^a,\bar{V}^b)$, where the first mapping is via (\ref{def.JagivenJb})--(\ref{def.JbgivenJa}) and the second is via (\ref{eq.2agentjointproblem}), has a fixed point.
As mentioned in the introduction, the main challenge of this approach is due to the presence of the ``floor" and ``ceiling" functions in the definition of $J^a$ and $J^b$, which introduces discontinuity. We propose the following general methodology for addressing this difficulty (which can be applied in many problems with similar types of discontinuities). This approach is based on the observation that the ``floor" and ``ceiling" operators can be made continuous, in the appropriate sense, if restricted to a set of functions that are ``sufficiently noisy" (see \cite{Delarue2}, \cite{NadShk}, for another application of this observation). Notice that $\lfloor v(X) \rfloor$ is the instantaneous reward process for the optimal stopping problem of the seller. Denote the corresponding optimal stopping time by $\tau$. Next, consider a small perturbation $\tilde{v}$ of $v$ in the uniform topology.
Assume that the process $v(X)$ is sufficiently noisy, in the sense that, for any $\varepsilon>0$, there exists a stopping time $\tau_{\varepsilon}$, such that $v(X_{\tau_{\varepsilon}})\geq v(X_{\tau})+\varepsilon$ and $\tau_{\varepsilon}\downarrow\tau$ as $\varepsilon\downarrow0$ (e.g., Brownian motion satisfies this property). Then, $\lfloor\tilde{v}(X_{\tau_{\varepsilon}})\rfloor\geq\lfloor v(X_{\tau})\rfloor$ holds, provided $\|\tilde{v}-v\|<\varepsilon$. Since the running rewards are continuous in time and $\tau_{\varepsilon}$ is close to $\tau$, we can find a stopping time for the perturbed problem that produces the objective value close to the optimal objective value of the unperturbed problem. Interchanging the two problems, we conclude that the value functions of the two problems are close, which yields the desired continuity of the mapping $(v^a,v^b)\mapsto (J^a,J^b)\mapsto (\bar{V}^a,\bar{V}^b)$.
Then, restricting the latter mapping to a compact, on which $v^{a/b}(X)$ are guaranteed to be sufficiently noisy, we can apply Schauder's theorem to obtain a fixed point.

We emphasize that this strategy is quite general and can be applied outside the scope of the present problem. However, in order to apply Schauder's theorem, one needs to show that the chosen compact is preserved by the mapping: in particular, that sufficiently noisy $v^{a/b}(X)$ are mapped into sufficiently noisy $\bar{V}^{a/b}(X)$. The proof of the latter may depend on the structure of a problem at hand.
In the present case, we provide such a proof in Section \ref{sec.individualvalfuncprops}, by utilizing the geometric approach of \cite{DayanikKaratzas}, \cite{Dayanik}, which can be applied to rather irregular linear diffusion stopping problems. It allows us to show that $\bar{V}^{a/b}(\cdot)$ have derivatives bounded away from zero, provided $(v^a(\cdot),v^b(\cdot))$ satisfy this property and the volatility $\sigma$ of the underlying process $X$ is sufficiently large (Proposition \ref{prop.Vmonotonicity}).\footnote{As follows from Proposition  \ref{prop.Vmonotonicity}, $\sigma$ has to be sufficiently large to ensure that $w(\sigma)<1$, where $w$ is defined in Proposition \ref{prop.fabderbounds}.} The latter, in turn, implies that $\bar{V}^{a/b}(X)$ are sufficiently noisy, in the above sense, and Proposition \ref{prop.VcontinuousinJ} uses this conclusion to prove the continuity of $(v^a,v^b)\mapsto (J^a,J^b)\mapsto (\bar{V}^a,\bar{V}^b)$.

Most of Section \ref{sec.responsepricesandequilibrium} is concerned with the characterization of the optimal continuous controls in a feedback form (Proposition \ref{prop.Va0=Vatrue}) and showing the regularity of the feedback functionals. Typically, such characterization is obtained via PDE or BSDE methods, but the latter cannot be applied in the present case (i.e. the well-posedness of the associated equations is not known) due to irregularity of the stopping reward processes. The proof of Theorem \ref{thm:main} is concluded by combining Proposition \ref{prop.Va0=Vatrue} and Theorem \ref{thm.fixedpointexistence}.


\section{Agents' optimal stopping problems}
\label{sec.individualvalfuncprops}

In this section, we investigate the properties of an individual agent's value function, given that the continuous controls of both agents, and the other agent's value function, are fixed. In other words, we consider the value function of an individual optimal stopping problem.
We first establish its basic relative boundedness and quasi-periodicity properties. Then, we recall the analytic machinery of second order ODEs associated with linear diffusions, which, together with the geometric approach to optimal stopping of \cite{DayanikKaratzas}, \cite{Dayanik}, allows us to establish a sufficiently strong monotonicity of the value function of an agent's optimal stopping problem. Finally, the established monotonicity, relative boundedness, and quasi-periodicity, allow us prove the continuity of the value function with respect to the value function of the other agent (for a fixed continuous control), which is the main result of this section, stated in Proposition \ref{prop.VcontinuousinJ}.


\subsection{Preliminary constructions}

In view of Assumption \ref{ass.Foverfmono} and the definition of $(\mathcal{A}^a,\mathcal{A}^b)$, it is no loss of generality to assume that $p^a\in\mathcal{A}^a$ and $p^b\in\mathcal{A}^b$ satisfy 
\begin{equation}\label{eq.C.def}
\left\Vert p^a(x)-x\right\Vert_{\infty}\le C,\quad\left\Vert p^b(x)-x\right\Vert_{\infty}\le C,
\end{equation}
with some fixed constant $C>0$ (depending only on $(c_l,\lambda,F)$), where $\Vert\cdot\Vert_\infty$ denotes $\mathbb{L}^\infty$ norm (i.e. the agents' maximum objective values would not change with such a restriction).
In addition, for any $p^a\in\mathcal{A}^a$ and $p^b\in\mathcal{A}^b$, we have:
\begin{equation}\label{eq.cgeqcl}
c(p^a(x),p^b(x),x)\ge c_l,\quad\forall x\in\RR.
\end{equation}
From the definition of $c$, and as $\lambda$ is fixed throughout the paper, we obtain
\begin{equation*}
c(p^a(x),p^b(x),x)\le c_u=2\lambda>0.
\end{equation*}
Another property we will frequently use is the following.
\begin{definition}\label{def.Cclosetox}
A measurable function $f:\RR\rightarrow\RR$ is ``\emph{$C$-close to $x$}", if
$$
\Vert f(x)-x\Vert_{\infty}\le C,\quad \forall\,x\in\RR,
$$
with the constant $C$ appearing in (\ref{eq.C.def}).
\end{definition}
Using the above property we define the notion of admissible barrier: i.e. function $v$, such that the reward function in the associated optimal stopping problem arises as a rounding of $v$.
\begin{definition}\label{def.admissiblebarriers}
A function $v:\RR\rightarrow\RR$ is an \emph{admissible barrier} if it is measurable and $C$-close to $x$.
\end{definition}
Next, we introduce the value functions of the optimal stopping problems faced by the agents:
\begin{equation}\label{def.VagivenJb}
V^a\left(x,p^a,p^{b},v\right):=\sup_{\tau}J^a\left(x,\tau,p^a,p^{b},v\right),\quad\forall\,x\in\RR,
\end{equation}
\begin{equation}\label{def.VbgivenJa}
V^b\left(x,p^{a},p^b,v\right):=\inf_{\tau}J^b\left(x,\tau,p^a,p^{b},v\right),\quad\forall\,x\in\RR,
\end{equation}
for $p^a\in\mathcal{A}^a$, $p^b\in\mathcal{A}^b$, and admissible barrier $v$.
It is easy to see these value functions are well-defined and finite.
Note that $p^{a}$, $p^b$, $\bar{V}^a$, $\bar{V}^b$ and $x$ are measured in ticks, and only the relative measurements are interpretable. Hence the following ansatz for an equilibrium is natural:
\begin{equation}
\bar{V}^a(x+1)=\bar{V}^a(x)+1,\,\, \bar{V}^b(x+1)=\bar{V}^b(x)+1,\,\,
p^{a}(x+1)=p^{a}(x)+1,\,\, p^{b}(x+1)=p^{b}(x)+1.
\end{equation}
Let us introduce a special term for the above property.
\begin{definition}
We say a function $f:\RR\rightarrow\RR$ has ``\emph{1-shift property}" if
$$
f(x+1)=f(x)+1,\quad\forall x\in\mathbb{R}.
$$
\end{definition}

We will also need a version of the objective and the associated value function without floors and ceilings in the reward function, and with $x$ subtracted (this will be particularly important for certain approximation arguments in Section \ref{sec.responsepricesandequilibrium}). Hence, for any $p^a\in\mathcal{A}^a$, $p^b\in\mathcal{A}^b$, any stopping time $\tau$, and any admissible barrier $v$, we introduce:
\begin{multline}\label{def.Ja0}
J^a_0(x,\tau,p^a,p^b,v)=\\\EE^x\Big[\int_0^\tau \exp\left(-\int_0^t c(p^a(X_s),p^b(X_s),X_s)\text{d}s\right) \left(g^a(p(X_t),p^b(X_t),X_t)-c(p^a(X_t),p^b(X_t),X_t)X_t
\right)
\text{d}t \\ +\exp\left(-\int_0^\tau c(p^a(X_s),p^b(X_s),X_s)\text{d}s\right)\left(v(X_\tau)-X_\tau\right)\Big],
\end{multline}
\begin{multline}\label{def.Jb0}
J^b_0(x,\tau,p^a,p^b,v)=\\\EE^x\Big[\int_0^\tau \exp\left(-\int_0^t c(p^a(X_s),p^b(X_s),X_s)\text{d}s\right) \left(g^b(p^a(X_t),p^b(X_t),X_t)-c(p^a(X_t),p^b(X_t),X_t)X_t
\right)
\text{d}t \\ +\exp\left(-\int_0^\tau c(p^a(X_s),p^b(X_s),X_s)\text{d}s\right)\left(v(X_\tau)-X_\tau\right)\Big],
\end{multline}
\begin{equation}\label{def.Va0}
V^a_0(x,p^a,p^b,v)=\sup_{\tau} J^a_0(x,\tau,p^a,p^b,v),
\end{equation}
\begin{equation}\label{def.Vb0}
V^b_0(x,p^a,p^b,v)=\inf_{\tau} J^b_0(x,\tau,p^a,p^b,v).
\end{equation}

Ultimately, we aim to show that there exists a solution of (\ref{eq.2agentjointproblem}), $(p^a,p^b,\bar{V}^a,\bar{V}^b)$, such that $(\bar{V}^a,\bar{V}^b)$ satisfy the 1-shift property and are $C$-close to $x$.
An important first step, then, is to verify that these properties are preserved by the mappings $v\mapsto V^a\left(\cdot,p^a,p^{b},v\right), V^b\left(\cdot,p^a,p^{b},v\right)$.
The following lemma shows that the 1-shift property is preserved. Its proof follows easily after rewriting the objective in the form (\ref{eq.dexpformofobjective}).

\begin{lemma}\label{lemma.shiftproperty}
Assume that $p^a\in\mathcal{A}^a$, $p^b\in\mathcal{A}^b$, and an admissible barrier $v$, have 1-shift property.
Then, so do $V^a(\cdot,p^a,p^b,v)$ and $V^b(\cdot,p^a,p^b,v)$, while $V^a_0(\cdot,p^a,p^b,v^b)$ and $V^b_0(\cdot,p^a,p^b,v^a)$ are 1-periodic.
In addition, we have
\begin{equation*}
\begin{split}
c(x+1,p^a+1,p^b+1)=c(x,p^a,p^b),\\
\frac{g^{a/b}}{c}(x+1,p^a+1,p^b+1)=\frac{g^{a/b}}{c}(x,p^a,p^b)+1,\\
\end{split}
\end{equation*}
for all $x\in\RR$, which means that $c(\cdot,p^a(\cdot),p^b(\cdot))$ is 1-periodic and $(g^{a/b}/c)(\cdot,p^a(\cdot),p^b(\cdot))$ have the 1-shift property.
\end{lemma}

In what follows, we will often suppress the dependence on $p^a(x)$, $p^b(x)$ from notation to avoid clutter. Hence, we denote
\begin{equation}\label{def.cgagbetcgivenpab}
\begin{split}
c(x):=c_p(x):=c(p^a(x),p^b(x),x),&\\
g^{a}(x):=g^{a}_p(x):=g^a(p^a(x),p^b(x),x),&
\quad g^b(x):=g^b_p(x):=g^b(p^a(x),p^b(x),x),\\
\mathcal{F}^b(x):= \mathcal{F}^b_p(x):=\int_{-\infty}^{p^b(x)-x} \lfloor x+\alpha y\rfloor \text{d}F(y),&
\quad \mathcal{F}^a(x):=\mathcal{F}^a_p(x):=\int_{p^a(x)-x}^{\infty} \lceil x+\alpha y \rceil \text{d}F(y),
\end{split}
\end{equation}
where we use the subscript $p$ whenever we want to emphasize the dependence on $p^a$, $p^b$.
The next lemma shows that the ``$C$-close to $x$" property is preserved.

\begin{lemma}\label{lemma.Cfromx}
Assume that $v$ is an admissible barrier, and that $p^a\in\mathcal{A}^a$, $p^b\in\mathcal{A}^b$ are such that $g^{a/b}/c$ are $C$-close to $x$. 
Then, $V^a(\cdot,p^a,p^b,v^b)$ and $V^b(\cdot,p^a,p^b,v^a)$ are $C$-close to $x$, while $V^a_0(\cdot,p^a,p^b,v^b)$ and $V^b_0(\cdot,p^a,p^b,v^a)$ are absolutely bounded by $C$.
\end{lemma}

\begin{proof}
We show the claim for $V^a$, the other ones being analogous. From the definitions (\ref{def.VagivenJb}) and (\ref{def.JagivenJb}) we obtain:
\begin{multline}\label{eq.dexpformofobjective}
V^a(x,p^a,p^b,v^b)-x=\sup_{\tau}\EE^x\left[
\int_0^\tau \exp\left(-\int_0^t c(X_s)\text{d}s\right)g^a(X_t)\text{d}t
+\exp\left(-\int_0^\tau c(X_s)\text{d}s\right)\lfloor v^b\rfloor
\right]-x\\
=\sup_{\tau}\EE^x\left[
\int_0^\tau \left(\frac{g^a}{c}(x)-x\right)\, \text{d}\left(-\exp\left(-\int_0^t c(X_s)\text{d}s\right)\right)+\exp\left(-\int_0^\tau c(X_s)\text{d}s\right)\left(\lfloor v^b\rfloor-x\rfloor\right)
\right].
\end{multline}
To get the upper bound, note that the right hand side of the above is
\begin{equation*}
\le 
\sup_{\tau}\EE^x\left[
\int_0^\tau C \text{d}\left(-\exp\left(-\int_0^t c(X_s)\text{d}s\right)\right)+\exp\left(-\int_0^\tau c(X_s)\text{d}s\right)C
\right]=C,
\end{equation*}
where we made use of $\frac{g^a}{c}-x\le C$, $\lfloor v^b\rfloor -x\le v^b-x\le C$, by the assumption of the lemma.
To obtain a lower bound, note that the same expression also is 
\begin{equation*}
\ge \EE^x\left[
\int_0^\infty \frac{g^a}{c}(x)-x \text{d}\left(-\exp\left(-\int_0^t c(X_s)\text{d}s\right)\right) \right]\ge
\EE^x\left[
\int_0^\infty -C \text{d}\left(-\exp\left(-\int_0^t c(X_s)\text{d}s\right)\right) \right]\ge -C,
\end{equation*}
where we used $\tau=\infty$. The claim for $V^a_0$ and $V^b_0$ is even simpler and can be proven as in the first part of this lemma.
\qed
\end{proof}

In what follows, we analyze $V^a(x,p^a,p^b,v)$ more closely ($V^b(x,p^a,p^b,v)$ being analogous). Ultimately, in Subsection \ref{subse:mon.cont}, we establish its monotonicity in $x$ and continuity in $v$, under appropriate conditions. Throughout this analysis, we think of $p^a$ and $p^b$ as fixed functions of $x$, while we vary $x$ and $v^b$($v^a$).

\subsection{Representation of the value function}

In this subsection, we develop a convenient representation of the value function $V^a(x,p^a,p^b,v)$ (cf. Proposition \ref{prop.vabmcmcharacterization}), which, along with Proposition \ref{prop.fabderbounds}, is needed to prove the main results of Subsection \ref{subse:mon.cont}. We will make heavy use of the well-known connection between linear diffusions and second-order ODEs. Our discounting and running cost functions are a bit less regular (measurable and locally bounded, but not continuous) than is commonly assumed in the literature, hence, a modicum of care is required to make this connection rigorous.

First, as in \cite{IMC}, for any given $p^a\in\mathcal{A}^a$, $p^b\in\mathcal{A}^b$, we define
\begin{equation}\label{def.psiprob}
\psi(x):=\psi_p(x):=\begin{cases}
\EE^x\left[\exp\left(-\int_0^{\tau_0} c_p(X_s)\text{d}s\right)\right] ,\, x\le 0,\\
\EE^0\left[\exp\left(-\int_0^{\tau_x} c_p(X_s)\text{d}s\right)\right]^{-1} ,\, x>0,
\end{cases}
\end{equation}
and
\begin{equation}\label{def.phiprob}
\phi(x):=\phi_p(x):=\begin{cases}
\EE^x\left[\exp\left(-\int_0^{\tau_0} c_p(X_s)\text{d}s\right)\right] ,\, x> 0,\\
\EE^0\left[\exp\left(-\int_0^{\tau_x} c_p(X_s)\text{d}s\right)\right]^{-1} ,\, x\le0,
\end{cases}
\end{equation}
where $\tau_x$ is the first hitting time of $x\in\RR$ by the process $X$.
It is clear that $\psi(0)=\phi(0)=1$, $\psi$ is strictly increasing, and $\phi$ is strictly decreasing. The results from \cite{IMC} (and the absolute continuity of the killing measure of the diffusion, in the present case) imply that $f=\phi,\psi$ has right derivative, $f^+$, defined everywhere, and it satisfies
$$
\frac{2}{\sigma^2}\int_{(a,b]}c(x)f(x)\text{d}x=f^+(b)-f^+(a),
$$
for all $b>a$. Passing to the limit as $b\downarrow a$ and $a\uparrow b$ shows that $f^+$ is continuous. One can also show the following.
\begin{lemma}\label{lemma.rightdertoC1}
If $f$ is continuous and has continuous right derivative on $[a,b]$ then $f\in C^1(a,b)$.
\end{lemma}
The proof of this fact is elementary and, hence, is omitted. Thus, the equation for $f=\phi,\psi$ can be rewritten as
\begin{equation}\label{eq.phipsiprimedifference}
\frac{2}{\sigma^2}\int_{(a,b]}c(x)f(x)\text{d}x=f'(b)-f'(a)
\end{equation}
As $c\in\mathbb{L}^\infty(\mathbb{R})$ and $f\in C(\mathbb{R})$, $f_{xx}$ exists a.e. and satisfies
$$
\frac{\sigma^2}{2}f_{xx}-cf=0,\,\text{a.e.},
$$
for $f=\psi,\phi$, and, in particular, $f$ belongs to the Sobolev space $\mathbb{W}^{2,loc}(\RR)$.
Next, we define
\begin{equation}\label{def.fa}
f^a(x):=f^a_p(x):=\frac{2}{\sigma^2 W}\left(
\phi(x)\int_{-\infty}^{x}\psi(y)g^a(y)\text{d}y+
\psi(x)\int_{x}^{\infty}\phi(y)g^a(y)\text{d}y
\right),
\end{equation}
\begin{equation}\label{def.fb}
f^b(x):=f^b_p(x):=-\frac{2}{\sigma^2 W}\left(
\phi(x)\int_{-\infty}^{x}\psi(y)g^b(y)\text{d}y+
\psi(x)\int_{x}^{\infty}\phi(y)g^b(y)\text{d}y
\right),
\end{equation}
where the Wronskian $W=\psi'(x)\phi(x)-\phi'(x)\psi(x)$ is independent of $x$ and strictly positive. Using the fact that $x\mapsto\phi(x)\int_{-\infty}^{x}\psi(y)g^a(y)\text{d}y$ (along with other similar terms) is absolutely continuous, as a product of two absolutely continuous functions, we conclude that $f^a$ has continuous derivative 
$$
\frac{2}{\sigma^2 W}\left(
\phi'(x)\int_{-\infty}^{x}\psi(y)g^a(y)\text{d}y+
\psi'(x)\int_{x}^{\infty}\phi(y)g^a(y)\text{d}y
\right),
$$
which is, in turn, a.e. differentiable, and furthermore
\begin{equation}\label{eq.faode}
\frac{\sigma^2}{2}f^a_{xx}-cf^a=-g^a,\,\text{a.e.}
\end{equation}
Similar claims hold for $f^b$. This, in particular, implies that $f^a$ and $f^b$ belong to $\mathbb{W}^{2,loc}(\RR)$. Applying Dynkin's formula, together with some obvious asymptotic properties of $\exp\left(-\int_0^t c(X_s)\text{d}s\right)f^a(X_t)$, we pass to the limit along a sequence of increasing to infinity stopping times, to obtain the following probabilistic representation
\begin{equation}\label{eq.faprobrep}
f^a(x)=\EE^x\left[\int_0^\infty \exp\left(-\int_0^t c(X_s)\text{d}s\right)g^a(X_t)\text{d}t\right].
\end{equation}
Similar representation holds for $f^b$.
Finally, we establish the following elementary bounds on $\phi$, $\psi$ (whose proof is also omitted).
\begin{lemma}\label{lemma.phipsiasymptotics}
Let $c_l>0$ ($c_u>0$) be the lower (resp. upper) bound of function $c$. Then, for all $x\in\RR$,
$$
\psi(x)\le \exp\left(\sqrt{\frac{2c_l}{\sigma^2}}x\right)\vee\exp\left(\sqrt{\frac{2c_u}{\sigma^2}}x\right)
$$
and
$$
\phi(x)\le \exp\left(-\sqrt{\frac{2c_l}{\sigma^2}}x\right)\vee
\exp\left(-\sqrt{\frac{2c_u}{\sigma^2}}x\right).
$$
\end{lemma}
Next, we show that $f^a$ and $f^b$ inherit the $C$-closeness to $x$ and 1-shift properties from $p^a$ and $p^b$.
\begin{lemma}\label{lemma.fabclosetox}
Assume that $p^a\in\mathcal{A}^a$ and $p^b\in\mathcal{A}^b$ are such that $g^{a/b}/c$ are $C$-close to $x$.
Then, so are $f^a$ and $f^b$.
\end{lemma}
\begin{proof}
The claim about $f^a$ follows from
\begin{equation}\label{eq.fadexprep}
f^a(x)-x=\EE^x\left[
\int_0^\infty \left(\frac{g^a}{c}(X_t)-X_t\right)
\text{d}\left(
\exp\left(-\int_0^t c(X_s)\text{d}s\right)
\right)
\right].
\end{equation}
Similarly for $f^b$.
\qed
\end{proof}
\begin{lemma}\label{lemma.fafbshiftproperty}
If $p^a\in\mathcal{A}^a$ and $p^b\in\mathcal{A}^b$ have 1-shift property, then so do $f^a$ and $f^b$.
\end{lemma}
\begin{proof}
Follows from the representation (\ref{eq.fadexprep}) and Lemma \ref{lemma.shiftproperty}.\qed
\end{proof}

It turns out that the 1-shift property of $p^a$ and $p^b$ implies that $f^a$ and $f^b$ converge to $f_0:x\mapsto x$ in $C^1$-norm, as the volatility $\sigma$ of the underlying diffusion $X$ increases to infinity. In particular, we can establish the following two-sided estimate on the derivatives of $f^a$ and $f^b$.

\begin{proposition}\label{prop.fabderbounds}
Assume that $p^a\in\mathcal{A}^a$ and $p^b\in\mathcal{A}^b$ have 1-shift property, and are such that $g^{a/b}/c$ are $C$-close to $x$.
Then, there exists a function $(C,c_l,c_u,\sigma)\mapsto w(C,c_l,c_u,\sigma)$, such that
\begin{equation}
1-w\le (f^{a/b})'(x)\le 1+w,\quad \forall\, x\in\RR,
\end{equation}
and
\begin{equation}
w(\sigma)\to 0,\,\,\text{as}\,\,\sigma\to\infty.
\end{equation}
\end{proposition}

\begin{proof}
We only show the upper bound on the derivative of $f^a$, the proof of the other parts being analogous.
Differentiating the representation (\ref{def.fa}), we obtain
\begin{multline}
\left(f^a\right)'(x)=
\frac{|\phi'(x)|\psi'(x)}{W}
\Big(
\int_x^\infty \frac{g^a}{c}(y) \phi(y)c(y)\frac{2}{\sigma^2 |\phi'(x)|}\text{d}y\,-\,
\int_{-\infty}^x \frac{g^a}{c}(y) \psi(y)c(y) \frac{2}{\sigma^2\psi'(x)}\text{d}y
\Big).
\end{multline}
As our diffusion $X$, killed at the rate $c(x)-c_l$, has $\pm\infty$ as natural boundary points, the results of \cite{IMC} imply $\phi'(-\infty)=0$, $\psi'(\infty)=0$. Hence, passing to the appropriate limits in (\ref{eq.phipsiprimedifference}), we obtain
\begin{equation}\label{eq.psi'integralrep}
\psi'(x)=\frac{2}{\sigma^2}\int_{-\infty}^x \psi(y)c(y)\text{d}y
\end{equation}
and 
\begin{equation}\label{eq.phi'integralrep}
\phi'(x)=-\frac{2}{\sigma^2}\int_{x}^\infty \phi(y)c(y)\text{d}y.
\end{equation}
From the representations above, we see that the function $y\mapsto2\phi(y)c(y)/(\sigma^2 |\phi'(x)|)$ is a probability density on $[x,\infty)$, and that $y\mapsto 2\psi(y)c(y) /(\sigma^2\psi'(x))$ is a probability density on $(-\infty,x]$.
Using this and the inequality $g^a(x)c(x)\le x+C$, we obtain
\begin{equation*}
\left(f^a\right)'(x)\le 2C\frac{|\phi'(x)|\psi'(x)}{W}+\left(\tilde{f}^a\right)'(x).
\end{equation*}
where $\tilde{f}^a(x)$ is defined as $f^a(x)$, with $g^a(x)$ replaced by $\tilde{g}^a(x):=xc(x)$. In particular, $f=\tilde{f}^a$ is the unique solution of
\begin{equation*}
\frac{\sigma^2}{2}f''(x)-c(x)f(x)=-xc(x),
\end{equation*}
which is easily seen to be given by $\tilde{f}^a(x)=x$. Hence, $\left(\tilde{f}^a\right)'=1$, and we conclude:
$$
\left(f^a\right)'(x)\le 1+2C\frac{|\phi'(x)|\psi'(x)}{W}.
$$
Thus, to prove the claim, it only remains to establish the appropriate bound on the last summand in the above.
To this end, we, first, notice that $c(x+1)=c(x)$, due to Lemma \ref{lemma.shiftproperty}, which implies: $\phi(x+1)=\gamma \phi(x)$ and $\psi(x+1)=\frac{1}{\gamma}\psi(x)$, with $0<\gamma=\phi(1)<1$. Hence, $x\mapsto|\phi'(x)|\psi'(x)$ is $1$-periodic, and it suffices to estimate $2C|\phi'(x)|\psi'(x)/W$ for $x\in[0,1]$.
To this end, recall that 
\begin{equation*}
|\phi'(x)|=\frac{2}{\sigma^2}\int_x^\infty c(y)\phi(y)\text{d}y.
\end{equation*}
For $x\ge0$, using Lemma \ref{lemma.phipsiasymptotics} and the bounds on $c$, we conclude that the right hand side of the above expression is bounded from above by
$$
\frac{2c_u}{\sigma^2} \int_x^\infty \exp\left(-\sqrt{\frac{2c_l}{\sigma^2}}y\right)\text{d}y=\frac{2 c_u}{\sigma \sqrt{2c_l}}\exp\left(-\sqrt{\frac{2c_l}{\sigma^2}}x\right).
$$
Combined with a similar estimate from below, the above yields:
\begin{equation}
\frac{2 c_l}{\sigma \sqrt{2c_u}}\exp\left(-\sqrt{\frac{2c_u}{\sigma^2}}x\right)\le|\phi'(x)|\le \frac{2 c_u}{\sigma\sqrt{2c_l}}\exp\left(-\sqrt{\frac{2c_l}{\sigma^2}}x\right),
\end{equation}
for $x\ge0$. Similarly, for $x\le0$, we obtain
$$
\frac{2 c_l}{\sigma \sqrt{2c_u}}\exp\left(\sqrt{\frac{2c_u}{\sigma^2}}x\right)\le|\psi'(x)|\le \frac{2 c_u}{\sigma\sqrt{2c_l}}\exp\left(\sqrt{\frac{2c_l}{\sigma^2}}x\right),
$$
which, using $\psi'(x+1)=\frac{1}{\gamma}\psi'(x)$, gives
\begin{equation}
\frac{1}{\gamma}\frac{2 c_l}{\sigma \sqrt{2c_u}}\exp\left(\sqrt{\frac{2c_u}{\sigma^2}}(x-1)\right)\le|\psi'(x)|\le \frac{1}{\gamma}\frac{2 c_u}{\sigma\sqrt{2c_l}}\exp\left(\sqrt{\frac{2c_l}{\sigma^2}}(x-1)\right),
\end{equation}
for $x\in[0,1]$. Combining the above and replacing the exponential terms with their worst-case bounds, we conclude:
$$
|\phi'(x)|\psi'(x)\le\frac{1}{\gamma} \left(\frac{2c_u}{\sigma \sqrt{2c_l}}\right)^2,\quad x\in[0,1].
$$
Note that on $[0,1]$, we have $ \gamma\leq \phi \leq 1\leq \psi \leq 1/\gamma$, which, together with the above estimates of $|\phi'|$ and $\psi'$, gives, for $x\in[0,1]$,
\begin{multline}
W=\phi\psi'+\psi|\phi'|\ge
\gamma\cdot\left(\frac{1}{\gamma}\frac{2c_l}{\sigma\sqrt{2c_u}}
\exp\left(\sqrt{\frac{2c_u}{\sigma^2}}(x-1)\right)
\right)+
1\cdot \frac{2c_l}{\sigma \sqrt{2c_u}}\exp\left(-\sqrt{\frac{2c_u}{\sigma^2}}x\right)\ge\\
\frac{4c_l}{\sigma\sqrt{2c_u}}\exp\left(-\sqrt{\frac{2c_u}{\sigma^2}}\right),
\end{multline}
and, in turn,
$$
\frac{|\phi'(x)|\psi'(x)}{W}\le\frac{1}{\sigma}\frac{1}{\gamma}\left(\frac{2c_u}{\sqrt{2c_l}}\right)^2\frac{\sqrt{2c_u}}{4c_l}
\exp\left(\frac{\sqrt{2c_u}-\sqrt{2c_l}}{\sigma}\right),\quad x\in[0,1].
$$
Finally, we notice that, as $\psi(1)=1/\gamma$, the bound on $\psi$ in Lemma \ref{lemma.phipsiasymptotics} implies
$$
\frac{1}{\gamma}\le \exp\left(\sqrt{\frac{2c_u}{\sigma^2}}\right).
$$
This, together with the previous inequality, implies the desired upper bound on $2C |\phi'(x)|\psi'(x)/W$ and, in turn, on $\left(f^a\right)'(x)$. \qed
\end{proof}

Next, we establish the desired representation of the value function in terms of $f^a$ and $f^b$, using the results of \cite{Dayanik}.
To this end, we define 
\begin{equation}\label{def.F}
\Fr(x):=\Fr_p(x):=\frac{\psi(x)}{\phi(x)},\quad x\in\RR,
\end{equation}
and introduce the following transformation:
\begin{equation}\label{def.widehat}
\widehat{\cdot}\,\colon h\mapsto \widehat{h},\quad \widehat{h}(y)=\frac{h}{\phi}(\Fr^{-1}(y)),
\end{equation} 
mapping any function $h:\RR\rightarrow\RR$ into function $\widehat{h}:(0,\infty)\rightarrow\RR$. 
The following lemma, in particular, describes the shift-property of transformed functions. Its proof is trivial and, hence, is omitted.

\begin{lemma}\label{lemma.ydomainshiftprops}
For any $p^a\in\mathcal{A}^a$, $p^b\in\mathcal{A}^b$, satisfying the 1-shift property, we have
\begin{equation}
\phi(x+1)=\gamma\phi(x),\quad\psi(x+1)=\frac{1}{\gamma}\psi(x),\quad \Fr(x+1)=\frac{1}{\gamma^2}\Fr(x),\quad \forall\,x\in\RR,
\end{equation}
where $\phi$, $\psi$, $\Fr$ are given by (\ref{def.phiprob}), (\ref{def.psiprob}), (\ref{def.F}), and $\gamma=\phi(1)\in(0,1)$.
Therefore,
\begin{equation}
\hat{H}\left(\frac{y}{\gamma^2}\right)=\frac{1}{\gamma}\hat{H}(y),\quad \forall\,y>0,
\end{equation}
for any $1$-periodic function $H:\RR\rightarrow\RR$. In particular, by Lemma \ref{lemma.fafbshiftproperty}, the above property holds for $H=v-f^a, f^b-v$, with any $v:\RR\rightarrow\RR$ satisfying the 1-shift property.
\end{lemma}

Using the transformation defined by (\ref{def.widehat}), we obtain the following ``geometric" description of the value function of an individual stopping problem.

\begin{proposition}\label{prop.vabmcmcharacterization}
For any $p^a\in\mathcal{A}^a$, $p^b\in\mathcal{A}^b$, and any admissible barrier $v$, the functions $x\mapsto V^a_1(x):=x+V^a_0(x,p^a,p^b,v)$ and $x\mapsto V^b_1(x):=x+V^b_0(x,p^a,p^b,v)$ are uniquely determined by
\begin{equation}
\begin{split}
\widehat{V}^a_1(y)=\operatorname{mcm}\widehat{\left(v-f^a\right)}(y)+\widehat{f}^a(y),\\
\widehat{V}^b_1(y)=-\operatorname{mcm}\widehat{\left(f^b-v\right)}(y)+\widehat{f}^b(y),
\end{split}
\end{equation}
where $\operatorname{mcm}(f)$ denotes the smallest nonnegative concave majorant of a function $f$ (equal to infinity, if no such majorants exist).
Similarly, the individual agents' value functions are uniquely determined by
\begin{equation}
\begin{split}
\widehat{V^a(\cdot,p^a,p^b,v)}(y)=\operatorname{mcm}\widehat{\left(\lfloor v\rfloor-f^a\right)}(y)+\widehat{f}^a(y)\\
\widehat{V^b(\cdot,p^a,p^b,v)}(y)=-\operatorname{mcm}\widehat{\left(f^b-\lceil v\rceil\right)}(y)+\widehat{f}^b(y)
\end{split}
\end{equation}
\end{proposition}

\begin{proof}
We only prove the claim for $V^a_1$. First, we recall (\ref{def.Ja0}), to obtain
\begin{multline*}
J^a_0(x,\tau,p^a,p^b,v)+x=
\EE^x\left[
\int_0^\tau \exp\left(-\int_0^t c(X_s)\text{d}s\right) g^a(X_t)\text{d}t+
\exp\left(-\int_0^\tau c(X_s)\text{d}s\right)v(X_{\tau})
\right]\\
=\EE^x 
\Big[
\int_0^\infty \exp\left(-\int_0^t c(X_s)\text{d}s\right) g^a(X_t)\text{d}t\\
- \exp\left(-\int_0^\tau c(X_s)\text{d}s\right)
\int_0^\infty \exp\left(-\int_0^t c(X_{\tau+s})\text{d}s\right) g^a(X_{\tau+t})\text{d}t
+\exp\left(-\int_0^\tau c(X_s)\text{d}s\right) v(X_{\tau})
\Big]\\
=f^a(x)+\EE^x \left[\exp\left(-\int_0^\tau c(X_s)\text{d}s\right)
\left(v(X_{\tau}) -f^a(X_\tau)\right)\right],
\end{multline*}
where the last equality follows from (\ref{eq.faprobrep}) and from the strong Markov property of $X$.
Hence,
$$
V^a_0(x,p^a,p^b,v)+x=f^a(x)+\sup_{\tau}\EE^x \left[\exp\left(-\int_0^\tau c(X_s)\text{d}s\right)
\left(v(X_{\tau}) -f^a(X_\tau)\right)\right].
$$
As $v-f^a$ is measurable and locally bounded, the last term above (i.e., the value function of a pure stopping problem (with discounting)) has the claimed $\operatorname{mcm}$-characterization by Proposition 3.4 from \cite{Dayanik}.\qed
\end{proof}

\subsection{Continuity via monotonicity}
\label{subse:mon.cont}

In this subsection, we establish the sufficiently strong monotonicity property of $V^{a/b}(\cdot,p^a,p^b,v)$, which, in turn, allows us to show the continuity of the mappings $v\mapsto V^{a/b}(x,p^a,p^b,v)$. Note that, a priori, the latter mappings are not continuous, as they involve the \emph{discontinuous rounding operators}. However, as discussed at the beginning of the paper, these mappings do become continuous if the process $v(X)$ is ``sufficiently noisy". The latter property, in particular, can be deduced from the strict monotonicity of $v(\cdot)$. The following proposition shows that the desired monotonicity property in $x$ is preserved by the mappings $v\mapsto V^{a/b}(x,p^a,p^b,v)$.

\begin{proposition}\label{prop.Vmonotonicity}
Assume that $p^a\in\mathcal{A}^a$, $p^b\in\mathcal{A}^b$, and admissible barrier $v$, have 1-shift property and are such that $g^{a/b}/c$ are $C$-close to $x$. Then, $V(x)=V^{a/b}(x,p^a,p^b,v),\,x+V^{a/b}_0(x,p^a,p^b,v)$ is absolutely continuous, and its derivative satisfies:
\begin{equation}
|V'(x)-1|\le w,\,\,\text{a.e. }x\in\RR,
\end{equation}
with $w(\sigma)\to0$, as $\sigma\to\infty$, uniformly over all $(p^a,p^b,v)$ satisfying the above properties (assuming that the rest of the model parameters, $(\lambda,F,\alpha)$, are fixed). In particular, there exists $\epsilon>0$, s.t.
$$
V'(x)\ge\epsilon,\,\,\text{a.e. }x\in\RR,
$$
for all sufficiently large $\sigma>0$. 
\end{proposition}

\begin{proof}
We only prove the lower bound on the derivative of $V(x)=x+V^a_0(x,p^a,p^b,v)$, the other parts being similar.
Note that Proposition \ref{prop.vabmcmcharacterization} implies
$$
V(x)=f^a(x)+\phi(x)\operatorname{mcm}\left(\widehat{v -f^a}\right)(\Fr(x)),\quad\forall\, x\in\RR.
$$
Notice that $f^a,\phi,\Fr\in C^1(\mathbb{R})$, and the value of ``$\operatorname{mcm}$" above is absolutely continuous, as it is concave (note that it is also finite, as $v -f^a$ is absolutely bounded, due to Lemma \ref{lemma.fabclosetox}, which, in turn, implies that $\widehat{v -f^a}$ is bounded above by an affine function). Then $V$ is also absolutely continuous, and its (a.e. defined) derivative satisfies:
$$
V'(x)=(f^a)'(x)+\phi'(x)\operatorname{mcm}\left(\widehat{ v -f^a}\right)(\Fr(x))+\phi(x)\Fr'(x)\operatorname{mcm}\left(\widehat{ v -f^a}\right)'(\Fr(x)).
$$
From Proposition \ref{prop.fabderbounds}, we obtain $(f^a)'\ge 1-w$, with $w$ as in the statement of the proposition. Hence, we only need to show that, for a.e. $x\in\RR$,
\begin{equation}\label{eq.phimcmprime}
V'(x) - (f^a)'(x) = \phi'(x)\operatorname{mcm}\left(\widehat{ v -f^a}\right)(\Fr(x))+\phi(x)\Fr'(x)\operatorname{mcm}\left(\widehat{ v -f^a}\right)'(\Fr(x))\ge -\tilde{w}(\sigma),
\end{equation}
with $\tilde{w}$ having the appropriate asymptotic properties. Recall that, by Lemma \ref{lemma.shiftproperty}, $V(x)=V^a_0(x)+x$ is $1$-periodic, hence, it suffices to consider $x$ on any bounded interval of length at least $1$, as opposed to the entire real line. To simplify the notation, we denote
$$
h_0(y):=\widehat{ v -f^a}(y),\quad h(y):=\operatorname{mcm}\left(h_0\right)(y),\quad \forall\, y>0.
$$
Note that the assumed $1$-shift property of $v$, via Lemma \ref{lemma.fafbshiftproperty}, implies that $ v -f^a$ is $1$-periodic. The latter and Lemma \ref{lemma.ydomainshiftprops}, in turn, imply that
$$
h_0(y/\gamma^2)=h_0(y)/\gamma,\quad \forall\,y>0.
$$
It can be easily checked (using the scaling properties of $\operatorname{mcm}$) that the above property passes on to the minimal concave majorant $h(y)$.

Next, we define 
$$
\bar{\phi}(y):=\widehat{\phi^2}(y)=\phi(\Fr^{-1}(y))
$$
It is easy to check that $\left(\frac{\sigma^2}{2}\frac{\partial^2}{\partial x^2}-c\right)(\phi^2)>0$, and, hence, $\bar{\phi}$ is convex by the following lemma, which can be proven by a straightforward calculation.

\begin{lemma}\label{lemma.3.1}
Let $H\in \mathbb{W}^{2,loc}$ and $-\infty<x_1<x_2<\infty$ be such that $\frac{\sigma^2}{2}H_{xx}-cH>0,\,(\text{resp. }<0)$ a.e. on $(x_1,x_2)$.
Then, $\widehat{H}$ is convex (resp. concave) on $(y_1,y_2)$, with $y_i=\Fr(x_i)$. 
\end{lemma}

\noindent Furthermore, $\bar{\phi}$ is decreasing and satisfies $\bar{\phi}(y/\gamma^2)=\gamma \bar{\phi}(y)$. In the rest of the proof, we use the established properties of $\bar{\phi}$, $h_0$, and $h$, to prove the inequality in (\ref{eq.phimcmprime}).

Let us define
$$
\tilde{c}:=\sup_{y\in[1,1/\gamma^2]} h_0(y)\bar{\phi}(y).
$$
Note that, as $(h_0\bar{\phi})(y/\gamma^2)=(h_0\bar{\phi})(y)$, we obtain
$$
h_0(y)\le \tilde{c}/\bar{\phi}(y),
$$
for all $y>0$ (and not just for $y\in[1,1/\gamma^2]$), as follows from the definition of $\tilde{c}$.

First, let us assume that $\tilde{c}\le0$. We will show that, in this case, $h(y)\equiv0$, and $\tilde{w}=0$ gives the desired lower bound in (\ref{eq.phimcmprime}). Indeed, the constant function $0$ is a concave majorant of $h_0$ in this case. If there exists $y>0$ s.t. $h(y)=z<0$, then, as $h(y/\gamma^2)=h(y)/\gamma$, all points $(y/\gamma^{2k},z/\gamma^k)$, for integer $k$, lie on the graph of $h$. However, it is easy to see that the slope between two consecutive points in this family increases, if $z<0$, contradicting the concavity of $h$.

Having dealt with the simpler case of $\tilde{c}\le0$, we assume $\tilde{c}>0$ for the rest of the proof.
As 
$$
h_0(y)\bar{\phi}(y)=\left( v -f^a\right)(\Fr^{-1}(y)),
$$ 
and $ v$ and $f^a$ are $C$-close to $x$ (see Lemma \ref{lemma.fabclosetox}), we conclude that $\tilde{c}\le 2C$. 
Moreover, $1/\bar{\phi}(y)=\hat{1}$ (i.e., the ``$\,\widehat{\cdot}\,$" transform applied to a constant function $1$) is concave by the previous lemma, as 
$$
\left(\frac{\sigma^2}{2}\frac{d^2}{dx^2}-c\right)(1)<0.
$$ 
Hence, $\tilde{c}/\bar{\phi}$ is a concave majorant of $h_0$. It is also shown above that $h\geq0$. These observations imply
$$
0<h(y)\le \tilde{c}/\bar{\phi}(y),\quad\forall y\in(0,\infty).
$$
From the definition of $\tilde{c}$, we can find an infinite sequence of points $\{y_i\}$ in $[1,1/\gamma^2]$ s.t. $(h_0\bar{\phi})(y_i)\to\tilde{c}$. Let $y_*$ be any concentration point of this sequence. Then, from the continuity of the concave majorant $h$, and by $h\le \tilde{c}/\bar{\phi}$, we obtain
$$
h(y_*)=\tilde{c}/\bar{\phi}(y_*).
$$
Recall that we only need to establish (\ref{eq.phimcmprime}) on some $x$-interval of length $\ge1$. It is convenient to use the $x$-interval corresponding (via $\Fr^{-1}$) to the $y$-interval $[y_*,y_*/\gamma^2]$. Note that, as $y_*\in[1,1/\gamma^2]$, the resulting $x$-interval necessarily lies inside $[0,2]$.
Note also that
$$
\phi(x)=\bar{\phi}(\Fr(x)),\quad \phi'(x)=\bar{\phi}'(\Fr(x))\Fr'(x),
$$
and, hence, the left hand side of (\ref{eq.phimcmprime}) can be rewritten as 
\begin{equation}\label{eq.F.h.aux.est}
\Fr'(x)\left(
\bar{\phi}'(\Fr(x))h(\Fr(x))+\bar{\phi}(\Fr(x))h'(\Fr(x))
\right).
\end{equation}
As $\Fr'\geq0$, to estimate (\ref{eq.F.h.aux.est}) from below, for $x\in[0,2]$, we will derive an estimate for
\begin{equation}\label{eq.phi.h.aux.est}
\bar{\phi}'(y)h(y)+\bar{\phi}(y)h'(y),
\end{equation}
for $y\in[y_*,y_*/\gamma^2]\subset[1,1/\gamma^4]$.
As $\bar{\phi}'\leq 0$, $h>0$, and, as shown below, $h'>0$, we will estimate $\bar{\phi}$, $h'$ from below, and $h$ from above.
Clearly, $\bar{\phi}(y)\ge \gamma^2$, for $y\in [1,1/\gamma^4]$. In addition, as $h\leq \tilde{c}/\bar{\phi}$, we obtain $h(y)\le \tilde{c}/\gamma^2$, in the $y$-range we consider. 

To estimate $h'$ on $[y_*,y_*/\gamma^2]$, note that $h(y)$ coincides with $\tilde{c}/\bar{\phi}(y)$ at the endpoints of this interval, and $h\le\tilde{c}/\bar{\phi}$ on the entire interval. Then, as $\tilde{c}/\bar{\phi}$ is differentiable, we must have 
$$
\left(\frac{\tilde{c}}{\bar{\phi}}\right)'\left(\frac{y_*}{\gamma^2}\right)\le h'\left(\frac{y_*}{\gamma^2}\right),
$$
as, otherwise, we get a contradiction with the domination relationship between the two functions in the left neighborhood of $y_*/\gamma^2$. In the above, and in the rest of the argument, $h'(y)$ is understood as the left derivative at $y=y_*/\gamma^2$, as the right derivative at $y=y_*$, and as any element in the superdifferential at $y\in(y_*,y_*/\gamma^2)$. 
The last inequality, together with the concavity of $h$, implies that, for all $y\in[y_*,y_*/\gamma^2]$,
$$
h'(y)\ge\left(\frac{\tilde{c}}{\bar{\phi}}\right)'\left(\frac{y_*}{\gamma^2}\right)=-\frac{\tilde{c}}{\bar{\phi}^2(y_*/\gamma^2)}\bar{\phi}'(y_*/\gamma^2)\ge-\frac{\tilde{c}}{\gamma^2}\bar{\phi}'(y_*/\gamma^2).
$$
Note further that, as $y_*/\gamma^2\le y/\gamma^2$, for any $y\in[y_*,y_*/\gamma^2]$, and as $-\bar{\phi}'$ is nonnegative and decreasing (due to convexity and monotonicity of $\bar{\phi}$), we obtain
$$
-\frac{\tilde{c}}{\gamma^2}\bar{\phi}'(y_*/\gamma^2)\ge
-\frac{\tilde{c}}{\gamma^2}\bar{\phi}'(y/\gamma^2)=
-\tilde{c} \gamma \bar{\phi}'(y),\quad y\in[y_*,y_*/\gamma^2],
$$
where we also used $\bar{\phi}'(y/\gamma^2)=\gamma^3 \bar{\phi}'(y)$.

The above inequalities imply a lower bound for (\ref{eq.phi.h.aux.est}) in terms of $\bar{\phi}'$. 
Passing on to (\ref{eq.F.h.aux.est}), we obtain, for $x\in[\Fr^{-1}(y_*),\Fr^{-1}(y_*)+1]$:
\begin{multline}
\Fr'(x)\left(
\bar{\phi}'(\Fr(x))h(\Fr(x))+\bar{\phi}(\Fr(x))h'(\Fr(x))
\right)\label{eq.barphi.h.est.aux}\\
\ge
\Fr'(x)\left(\bar{\phi}'(\Fr(x))\frac{\tilde{c}}{\gamma^2}-\tilde{c}\gamma\bar{\phi}'(\Fr(x)) \right)=
\Fr'(x)\bar{\phi}'(\Fr(x))\tilde{c}\left(\frac{1}{\gamma^2}-\gamma^3\right).
\end{multline}
As $\Fr'(x)\bar{\phi}'(\Fr(x))=\phi'(x)\ge\phi'(0)$ and $\tilde{c}\le 2C$, the right hand side of the above is
$$
\ge \phi'(0)2C\left(\frac{1}{\gamma^2}-\gamma^3\right).
$$
The value of $\phi'(0)$ can be estimated using its integral representation, as in the proof of Proposition \ref{prop.fabderbounds}, using the asymptotic properties of $\phi$ from Lemma \ref{lemma.phipsiasymptotics}. As a result, we obtain:
$$
\phi'(0)\ge -\frac{2c_u}{\sigma\sqrt{2c_l}}.
$$ 
In addition, Lemma \ref{lemma.phipsiasymptotics} implies
$$ 
1\leq\exp\left(\sqrt{\frac{2c_l}{\sigma^2}}\right)\le\frac{1}{\gamma}\le\exp\left(\sqrt{\frac{2c_u}{\sigma^2}}\right).
$$  
Collecting the above, we conclude that the right hand side of (\ref{eq.barphi.h.est.aux}) vanishes, as $\sigma\to\infty$, at a rate depending only on $c_l$, $c_u$. This, in turn, yields (\ref{eq.phimcmprime}) and completes the proof of the proposition.
\qed
\end{proof}
 
Recall that our ultimate goal is to prove existence of a fixed point for a mapping which involves $v\mapsto V^{a/b}(x,p^a,p^b,v)$. As announced at the beginning of this subsection, in order to establish the continuity of the latter, we need to restrict $v$ to a set of sufficiently monotone functions. Proposition \ref{prop.Vmonotonicity} shows that the desired monotonicity property in $x$ is preserved by $v\mapsto V^{a/b}(x,p^a,p^b,v)$, and the following proposition, in turn, shows that the latter mapping is continuous on the set of such $v$.

\begin{proposition}\label{prop.VcontinuousinJ}
Assume that $p^a\in\mathcal{A}^a$, $p^b\in\mathcal{A}^b$, and admissible barriers $v^1$, $v^2$, have 1-shift property and are such that $g^{a/b}/c$ are $C$-close to $x$. Assume also that $(v^1)'(x),(v^2)'(x)\ge\epsilon>0$, for all $x\in\RR$. Then, there exists a function $\varepsilon:\RR_+\rightarrow\RR_+$ (depending only on $(\epsilon,\sigma,\lambda,F,\alpha)$, but independent of $(p^a,p^b,v^1,v^2)$, satisfying the above assumptions), such that $\varepsilon(\delta)\to0$, as $\delta\to0$, and
\begin{equation}
\sup_{x\in\RR} \left\vert V^{a/b}(x,p^a,p^b,v^1)-V^{a/b}(x,p^a,p^b,v^2)\right\vert\le \varepsilon\left(\sup_{x\in\RR}|v^1(x)-v^2(x)|\right).
\end{equation}
\end{proposition}

\begin{proof}
We will prove the statement for $V^a$, the one for $V^b$ being similar.
We will show that, whenever $\sup_{x\in\RR}|v^1(x)-v^2(x)|\leq \delta$, we have $V^a(x,p^a,p^b,v^1)\ge V^a(x,p^a,p^b,v^2)-\varepsilon(\delta)$, with $\varepsilon$ vanishing at zero. This, together with the symmetric inequality (proved analogously), yields the statement of the proposition.

For a given $\delta>0$, consider a $\delta$-optimal $\tau_2$, such that 
$$
J^a(\tau_2,x,p^a,p^b,v^2)\ge V^a(x,p^a,p^b,v^2)-\delta.
$$ 
Notice that it suffices to find $\tau_1$, such that 
$$
J^a(\tau_1,x,p^a,p^b,v^1)\ge J^a(\tau_2,x,p^a,p^b,v^2)-\varepsilon(\delta).
$$
Throughout this proof, $\varepsilon$ may change from line to line, but it always satisfies the properties stated in the proposition.
We construct $\tau_1\ge\tau_2$, separately, on two different $\mathcal{F}_{\tau_2}$-measurable sets.
On the event 
$$
\Omega_1:=\left\{\omega\colon\lfloor v^1(X_{\tau_2})\rfloor\ge \lfloor v^2(X_{\tau_2})\rfloor\right\},
$$ 
we set $\tau_1=\tau_2$.
If $\lfloor v^1(X_{\tau_2})\rfloor<\lfloor v^2(X_{\tau_2})\rfloor$, we, nevertheless, have 
$$
v^1(X_{\tau_2})\ge v^2(X_{\tau_2})-\delta,
$$
and, hence, by the assumption on $(v^1)'$,
$$
v^1\left(X_{\tau_2}+\frac{\delta}{\epsilon}\right)\ge v^2(X_{\tau_2}).
$$
The above implies
\begin{equation}\label{eq.J1attau10beatsJ2tau2}
\lfloor v^1\left(X_{\tau_2}+\delta/\epsilon\right) \rfloor\ge \lfloor v^2(X_{\tau_2})\rfloor.
\end{equation}
Then, on the event
$$
\Omega_2:=\Omega_1^c=\left\{\omega\colon \lfloor v^1(X_{\tau_2})\rfloor<\lfloor v^2(X_{\tau_2})\rfloor\right\},
$$ 
we define
$$
\tau_{10}:=\inf\left\{t\ge\tau_2\colon X_t\ge X_{\tau_2}+\frac{\delta}{\epsilon}\right\},
\quad \tau_{11}:=\inf\left\{t\ge\tau_2\colon X_t\le X_{\tau_2}-1\right\},
\quad \tau_1:=\tau_{10}\wedge\tau_{11}.
$$ 
In the subsequent derivations, we express various quantities in terms of the following expression, which can be interpreted as the ``relative to $x$" objective value, and which is more convenient than its ``absolute" version.
\begin{multline}\label{eq.Jreltoxformula}
J^a(\tau,x,p^a,p^b,v^i)-x=
\EE^x\left[
\int_0^\tau \exp\left(-\int_0^t c(X_s)\text{d}s\right)\left(g^a(X_t)-c(X_t)X_t\right)\text{d}t\right.\\
\left.+\exp\left(-\int_0^\tau c(X_s)\text{d}s\right)\left(\lfloor v^i\left(X_\tau\right)\rfloor -X_\tau\right)
\right],
\end{multline}
where $\left\vert \lfloor v^i(x)\rfloor -x\right\vert\le C+1$ and $|g^a(x)-c(x)x|\le c_u C$, by the assumption of the proposition.
Using the above expression, we obtain
\begin{multline*}
J^a(\tau_1,x,p^a,p^b,v^1)-J^a(\tau_2,x,p^a,p^b,v^2)=
\EE^x\Big[
\bone_{\Omega_1}\exp\left(-\int_0^{\tau_2}c(X_s)\text{d}s\right)\left(\lfloor v^1(X_{\tau_2})\rfloor-\lfloor v^2(X_{\tau_2})\rfloor\right)\\
+\bone_{\Omega_2}\int_{\tau_2}^{\tau_1} \exp\left(-\int_0^t c(X_s)\text{d}s\right)\left(g^a(X_t)-c(X_t)X_t\right)\text{d}t\\
+\bone_{\Omega_2}\left(
\exp\left(-\int_0^{\tau_1} c(X_s)\text{d}s\right)\left(\lfloor v^1\left(X_{\tau_1}\right)\rfloor -X_{\tau_1}\right)-
\exp\left(-\int_0^{\tau_2} c(X_s)\text{d}s\right)\left(\lfloor v^2\left(X_{\tau_2}\right)\rfloor -X_{\tau_2}\right)
\right)
\Big].
\end{multline*}
Note that the first one of the three summands, inside the above expectation, is nonnegative for every $\omega$, by the definition of $\Omega_1$.
Note also that, as $|g^a(x)-c(x)x|\le c_u C$, we have the following bound for the second summand:
\begin{multline*}
\left\vert\EE^x\left[
\bone_{\Omega_2}\int_{\tau_2}^{\tau_1} \exp\left(-\int_0^t c(X_s)\text{d}s\right)\left(g^a(X_t)-c(X_t)X_t\right)\text{d}t
\right]\right\vert\le
c_u C \EE^x |\tau_1-\tau_2|=c_u C \EE^0 \tau' =: \varepsilon(\delta),
\end{multline*}
where 
$$
\tau':=\inf\left\{t\ge0\colon X_t\notin(-1, \delta/\epsilon)\right\},
$$ 
and $\EE^0 \tau'$ is easily seen to go to zero as $\underline{\underline{O}}(\delta)$, as $\delta\to0$.
It only remains to estimate the expectation of the last summand,which can be decomposed as
\begin{multline*}
\EE^x\Big[\bone_{\Omega_2}\bone_{\{\tau_1=\tau_{10}\}}\left(
\exp\left(-\int_0^{\tau_{10}} c(X_s)\text{d}s\right)\left(\lfloor v^1\left(X_{\tau_{10}}\right)\rfloor -X_{\tau_{10}}\right)\right.\\
\left.\phantom{?????????????????????????}
-\exp\left(-\int_0^{\tau_2} c(X_s)\text{d}s\right)\left(\lfloor v^2\left(X_{\tau_2}\right)\rfloor -X_{\tau_2}\right)
\right)\\
+\bone_{\Omega_2}\bone_{\{\tau_1=\tau_{11}\}}\left(
\exp\left(-\int_0^{\tau_{11}} c(X_s)\text{d}s\right)\left(\lfloor v^1\left(X_{\tau_{11}}\right)\rfloor -X_{\tau_{11}}\right)-
\exp\left(-\int_0^{\tau_2} c(X_s)\text{d}s\right)\left(\lfloor v^2\left(X_{\tau_2}\right)\rfloor -X_{\tau_2}\right)
\right)
\Big].
\end{multline*}
As $|v^i(x)-x|\le C$, for all $x\in\RR$, and as 
$$
\PP^x\left(\tau_1=\tau_{11}\right)=\PP^0\left(\inf\left\{t\ge0\,:\, X_t=\frac{\delta}{\epsilon}\right\}>\inf\left\{t\ge0 \,:\,X_t=-1\right\}\right)=\frac{\delta/\epsilon}{1+\delta/\epsilon}=\underline{\underline{O}}(\delta),
$$ 
for $\delta\to0$, we obtain:
\begin{multline*}
\left\vert\EE^x\Big[\bone_{\Omega_2}\bone_{\{\tau_1=\tau_{11}\}}\left(
\exp\left(-\int_0^{\tau_{11}} c(X_s)\text{d}s\right)\left(\lfloor v^1\left(X_{\tau_{11}}\right)\rfloor -X_{\tau_{11}}\right)\right.\right.\\
\left.\left.-\exp\left(-\int_0^{\tau_2} c(X_s)\text{d}s\right)\left(\lfloor v^2\left(X_{\tau_2}\right)\rfloor -X_{\tau_2}\right)
\right)
\Big]\right\vert
\le 2(C+1)\PP^x(\tau_1=\tau_{11})=:\varepsilon(\delta).
\end{multline*}
Finally, we estimate the remaining term from below:
\begin{multline*}
\EE^x\Big[\bone_{\Omega_2}\bone_{\{\tau_1=\tau_{10}\}}\left(
\exp\left(-\int_0^{\tau_{10}} c(X_s)\text{d}s\right)\left(\lfloor v^1\left(X_{\tau_{10}}\right)\rfloor -X_{\tau_{10}}\right)\right.\\
\left.\phantom{?????????????????????????}
-\exp\left(-\int_0^{\tau_2} c(X_s)\text{d}s\right)\left(\lfloor v^2\left(X_{\tau_2}\right)\rfloor -X_{\tau_2}\right)
\right) \\
\geq\EE^x\Big[
\bone_{\Omega_2}\bone_{\{\tau_1=\tau_{10}\}}\Big(
\left(\lfloor v^2(X_{\tau_2})\rfloor-X_{\tau_2}\right)\left(\exp\left(-\int_0^{\tau_{10}} c(X_s)\text{d}s\right)-\exp\left(-\int_0^{\tau_{2}} c(X_s)\text{d}s\right)\right)\\
\phantom{???????????????}+\left(X_{\tau_2}-X_{\tau_{10}}\right)\exp\left(-\int_0^{\tau_{10}} c(X_s)\text{d}s\right)
\Big)
\Big]\\
\geq -(C+1)\EE^x\left[\exp\left(-\int_0^{\tau_{2}} c(X_s)\text{d}s\right)-\exp\left(-\int_0^{\tau_{1}} c(X_s)\text{d}s\right)\right]-\frac{\delta}{\epsilon},
\end{multline*}
where the first inequality follows from $\lfloor v^1(X_{\tau_{10}})\rfloor=\lfloor v^1(X_{\tau_2}+\delta/\epsilon)\rfloor\ge \lfloor v^2(X_{\tau_2})\rfloor$, by (\ref{eq.J1attau10beatsJ2tau2}), and the second inequality follows from $X_{\tau_2}-X_{\tau_{10}}=-\delta/\epsilon$ and $|\lfloor v^2(x)\rfloor -x|\le C+1$, together with $\tau_1\geq\tau_2$.
It only remains to notice that 
\begin{multline*}
\left\vert\EE^x\left[\exp\left(-\int_0^{\tau_{2}} c(X_s)\text{d}s\right)-\exp\left(-\int_0^{\tau_{1}} c(X_s)\text{d}s\right)\right]\right\vert\le\\
\EE^x \left\vert \int_{\tau_2}^{\tau_1} c(X_s)\text{d}s\right\vert\le
c_u \EE^x |\tau_1-\tau_2| = \underline{\underline{O}}(\delta),
\end{multline*}
which concludes the proof. \qed
\end{proof}


\section{Optimization over continous controls and existence of equilibrium}
\label{sec.responsepricesandequilibrium}

In this section, we, first, address the continuous control part of each agent's optimization problem. Namely, we introduce the feedback control operators and show that the controls they produce are indeed optimal. Our situation is somewhat less regular than the one treated in standard references, hence, we need to exploit the special structure of the problem and develop additional tricks to show this optimality. We, then, prove that these response control operators are continuous in the appropriate topology and show how the system of coupled optimization problems (\ref{eq.2agentjointproblem}) reduces to a fixed point problem of a certain mapping. Finally, we show the continuity of this mapping and the existence of its fixed point, satisfying the desired properties, which completes the proof of Theorem \ref{thm:main}.

For any measurable $v:\RR\rightarrow\RR$, we define the following feedback control operators:
\begin{equation}\label{def.responsecontrols}
\begin{split}
P^a(v)(x):=\min\operatorname{argmax}_{p\in\mathcal{A}^a(x)} \left(p-v(x)\right)F^+\left(p-x\right),\quad x\in\RR,\\
P^b(v)(x):=\max\operatorname{argmax}_{p\in\mathcal{A}^b(x)} \left(v(x)-p\right)F\left(p-x\right),\quad x\in\RR,
\end{split}
\end{equation}
where, for $x\in\RR$, we denote
$$
\mathcal{A}^a(x) := \{p\in\mathbb{Z}\,:\,1-F(p-x) \geq \frac{c_l}{2\lambda}\},
\quad \mathcal{A}^b(x) = \{p\in\mathbb{Z}\,:\,F(p-x) \geq \frac{c_l}{2\lambda}\},
\quad F^+(x):=1-F(x),
$$
with the c.d.f. $F$ (cf. (\ref{eq.Aa.Ab.def})). It is clear that, for a fixed $x\in\RR$, the set $\mathcal{A}^a(x)$ represents the possible values of a continuous control $p^a(x)$, and similarly for $\mathcal{A}^b(x)$.
It is also easy to see that, for any measurable $v$, the functions $P^a(v)$ and $P^b(v)$ are measurable, hence, they belong to $\mathcal{A}^a$ and $\mathcal{A}^b$, respectively (i.e. $P^a(v)$ and $P^b(v)$ are admissible continuous controls).
The following proposition, whose proof is given in the appendix, allows us to reduce the control-stopping problem of an agent to a fixed point problem, associated with optimal stopping and feedback control.

\begin{proposition}\label{prop.Va0=Vatrue}
Let $\sigma>0$ be sufficiently large, so that Proposition \ref{prop.fabderbounds} holds with $w<1$.
Consider any $p^b\in\mathcal{A}^b$ and any admissible barrier $v$, both satisfying the 1-shift property and such that $v'(x)\ge1-w>0$, for all $x\in\RR$.
Assume that there exists a measurable $\overline{V}^{a}:\RR\rightarrow\RR$, such that
\begin{equation}\label{eq.prop.Va0=Vatrue.eq1}
\overline{V}^{a}(x)=\sup_{\tau} J^{a}(x,\tau,P^{a}(\overline{V}^a),p^b,v),\quad\forall\,x\in\RR,
\end{equation}
and $\frac{g^a}{c}(P^{a}(\overline{V}^a)(\cdot),p^b(\cdot),\cdot)$ is $C$-close to $x$. Then,
\begin{equation}\label{eq.prop.Va0=Vatrue.eq2}
\overline{V}^a(x)=\sup_{p^a\in\mathcal{A}^a,\,\tau} J^a(x,\tau,p^a,p^b,v),\quad\forall\,x\in\RR.
\end{equation}
Analogous statement holds for $(\overline{V}^b, J^b,P^b(\overline{V}^b))$.
\end{proposition}

\begin{remark}
Note that it is easier to prove the converse statement: i.e., $\overline{V}^a$, defined by (\ref{eq.prop.Va0=Vatrue.eq2}), satisfies (\ref{eq.prop.Va0=Vatrue.eq1}). However, it does not imply the statement of the proposition, as there may exist multiple solutions to (\ref{eq.prop.Va0=Vatrue.eq1}). For the subsequent results, it is important to show that any solution to (\ref{eq.prop.Va0=Vatrue.eq1}) satisfies (\ref{eq.prop.Va0=Vatrue.eq2}).
\end{remark}

Proposition \ref{prop.Va0=Vatrue} allows us to sidestep the optimization over $p^a$ or $p^b$, in the definitions of $V^a$ and $V^b$, respectively, by using the feedback controls $P^a$ and $P^b$ throughout. 

Next, we notice that Assumption \ref{ass.Foverfmono}, in particular, implies that the optimal feedback prices are always $C'_0=C_0+1$-close to $x$, and also inherit the 1-shift property from the barriers they correspond to. This observation is formalized in the following lemma.

\begin{lemma}\label{lemma.responsecontrolprops}
Let
$$
p^a(x):=P^a(v)(x),\,\,p^b(x):=P^b(v)(x),\quad\forall\,x\in\RR,
$$
for some admissible barrier $v$. Then, for all $x\in\RR$,
$$
|p^a(x)-x|\le C'_0,\quad |p^b(x)-x|\le C'_0,
\quad \left\vert g^{a/b}(x)/c(x)\right\vert\le C'_0.
$$
If, in addition, $v$ has 1-shift property, then so do $p^a$ and $p^b$.
\end{lemma}

\begin{proof}
From the definition (\ref{def.responsecontrols}) and the fact that $\supp\,\xi\subset[-C_0,C_0]$, it is easy to see that $p^a(x)-x$ must be no smaller than the largest integer $\le -C_0$ and no larger than the smallest integer $\ge C_0$. Hence, 
$$
p^a(x)\ge x-C'_0,\quad p^a(x)\le x+C'_0.
$$
Similar conclusion holds for $p^b(x)$.
From (\ref{def.cpapbx}), (\ref{def.gapapbx}) we obtain
$$
\left\vert\frac{g^a}{c}(x)-x\right\vert=
\left\vert
\frac{(p^a(x)-x)\left(1-F(p^a(x)-x)\right)+\mathcal{F}^b(p^b(x),x)-xF(p^b(x)-x)}{\left(1-F(p^a(x)-x)\right)+F(p^b(x)-x)}
\right\vert.
$$
An analogous statement holds for $g^b$.
Thus, to prove the claim, it suffices to show
$$
|p^a(x)-x|\le C'_0,\quad |\mathcal{F}^b(p^b(x),x)-x|\le C'_0\, F(p^b(x)-x).
$$
The first inequality has already been established. For the second one, we have
$$
\mathcal{F}^b(p^b(x),x)-x=
\int_{-\infty}^{p^b(x)-x}\lfloor x+\alpha y\rfloor -x\, \text{d}F(y)
$$
To finish the proof, we notice that
$$
\left\vert \lfloor x+\alpha y\rfloor -x \right\vert\le C'_0,
$$
when $y\in\text{supp}\,\xi$ (as $\text{d}F(y)=0$ otherwise). The claim for $g^b/c$ can be proven analogously.
The 1-shift property of $p^a$ and $p^b$, given that $v$ satisfies it, is immediate from (\ref{def.responsecontrols}).
\qed
\end{proof}

In view of the above lemma, it is natural to choose $C'_0$ as the constant $C$, appearing in ``$C$-close to $x$" property. Note that $C'_0$ satisfies (\ref{eq.C.def}).

Next, for any admissible barriers $(v^a,v^b)$, we define
\begin{equation}\label{def.overlinePhi}
\overline{\Phi}(v^a,v^b)=\left(V^a\left(\cdot,P^a(v^a),P^b(v^b),v^b\right),\,
V^b(\cdot,P^a(v^a),P^b(v^b),v^a)\right).
\end{equation}
Lemmas \ref{lemma.Cfromx} and \ref{lemma.responsecontrolprops} imply that the components of $\overline{\Phi}$ are admissible barriers, and we can iterate this mapping. In fact, we are only interested in the restriction of $\overline{\Phi}$ to the sets $A_0$ and $A_0(w)$, defined below.

\begin{definition}\label{def.A0}
The set $A_0$ consists of all continuous real-valued functions on $\RR$, which are $C'_0$-close to $x$ and satisfy the 1-shift property. For any $w\geq0$, we say that $v\in A_0(w)$, if $v\in A_0$, $v$ is absolutely continuous, and $1-w\le v'\le 1+w$ a.e..
We equip $A_0$ and $A_0(w)$ with the topology of uniform convergence on all compacts.
\end{definition}

Lemmas \ref{lemma.shiftproperty}, \ref{lemma.Cfromx}, and \ref{lemma.responsecontrolprops}, show that $\overline{\Phi}$ maps $A_0\times A_0$ into itself. In addition, Proposition \ref{prop.Vmonotonicity} shows that $\overline{\Phi}$ maps $A_0\times A_0$ into $A_0(w)\times A_0(w)$, where $w$ can be chosen to be arbitrarily small for sufficiently large $\sigma>0$.
Using Proposition \ref{prop.Va0=Vatrue}, we show, below, that a fixed point of this mapping in the appropriate subset gives a solution to the system (\ref{eq.2agentjointproblem}). Hence, our next goal is to establish the existence of such a fixed point.
The first step is to show that $\overline{\Phi}$ is continuous on $A_0(w)\times A_0(w)$, for $w<1$. To this end, we, first, choose the appropriate space and topology for the feedback price controls $P^a(v)$ and $P^b(v)$, and show that they are continuous in $v\in A_0(w)$. Then, we show that $V^a\left(\cdot,p^a,p^b,v\right)$ and $V^b\left(\cdot,p^a,p^b,v\right)$ are continuous as operators acting on functions $(p^a,p^b)$, with respect to the chosen topology, uniformly in $v\in A_0$. This, together with the continuity of $V^a\left(\cdot,p^a,p^b,v\right)$ and $V^b\left(\cdot,p^a,p^b,v\right)$ in $v$, established in Proposition \ref{prop.VcontinuousinJ}, yields the continuity of $\overline{\Phi}$.

Let us define the space for the feedback price controls.
\begin{definition}
Denote by $B^a_0$ and $B^b_0$ the subspaces of $\mathcal{A}^a$ and $\mathcal{A}^b$, respectively, consisting of all functions that are $C'_0$-close to $x$ and satisfy the 1-shift property.
We equip $B^{a/b}_0$ with the topology induced by their natural restriction to $\mathbb{L}^1([0,1])$ (in view of the 1-shift property).
\end{definition}
Note that $P^{a}(v)\in B^a_0$ and $P^b(v)\in B^b_0$, for any $v\in A_0$.
The following, somewhat tricky, lemma is the first one of the two remaining results we need to establish the continuity of $\overline{\Phi}$. 

\begin{lemma}\label{lemma.PofVcontinuous}
For any $w\in[0,1)$, the mappings
$$
v\mapsto P^a(v),\quad v\mapsto P^b(v),
$$
from $A_0(w)$ into $B^a_0$ and $B^b_0$, respectively, are continuous.
\end{lemma}

\begin{proof}
We only show the $P^a$ version, the $P^b$ one being analogous. The proof consists of two steps. First, we show that, given $v$, with the properties described in the statement of the lemma (in particular, $v$ is increasing), $P^a(v)(x)$ is also an increasing function of $x$. Then, we use this monotonicity property to show the desired continuity of $P^a$.

{\bf Step 1.}  For a fixed $v$, we denote $p_x = P^a(v)(x)$. Assume, to the contrary, that for some $x_1>x_2$ we have $p_{x_1}<p_{x_2}$. Note that the set of admissible control values, $\mathcal{A}^a(x)$, shifts upward when $x$ increases. Therefore, if $p_{x_2}$ is an admissible control value at $x=x_2<x_1$, and $p_{x_2}>p_{x_1}$, with $p_{x_1}$ being an admissible control value at $x=x_1$, then $p_{x_2}$ is an admissible control value at $x=x_1$. Similarly, if $p_{x_1}<p_{x_2}$, then $p_{x_1}$ is admissible at $x=x_2$. Thus, to obtain a contradiction, it suffices to show that $p_{x_2}$ yields a better local objective value than $p_{x_1}$, at $x=x_1$: i.e.
\begin{equation}\label{eq.tmp1}
(p_{x_2}-v(x_1))F^+(p_{x_2}-x_1)>(p_{x_1}-v(x_1))F^+(p_{x_1}-x_1).
\end{equation}
The above inequality, clearly, holds if $p_{x_1}\le v(x_1)$. Hence, without loss of generality, we assume $p_{x_1}>v(x_1)$.
Then, (\ref{eq.tmp1}) is equivalent to:
\begin{equation}\label{eq.sect4.pZmonoratiodesired}
\frac{p_{x_2}-v(x_1)}{p_{x_1}-v(x_1)}>\frac{F^+(p_{x_1}-x_1)}{F^+(p_{x_2}-x_1)}
\end{equation}
Note that
\begin{equation}\label{eq.tmp2}
(p_{x_2}-v(x_2))F^+(p_{x_2}-x_2)\ge(p_{x_1}-v(x_2))F^+(p_{x_1}-x_2),
\end{equation}
due to the fact that $p_{x_2}$ is the optimal price at $x=x_2$ and, hence, is not worse than $p_{x_1}$ (which is admissible at $x=x_2$). The assumption $p_{x_1}>v(x_1)$ also implies that $p_{x_2}>v(x_2)$, as $v(x_2)<v(x_1)$. Hence, the inequality (\ref{eq.tmp2}) is equivalent to
\begin{equation}\label{eq.sect4.pZmonoratiogiven}
\frac{p_{x_2}-v(x_2)}{p_{x_1}-v(x_2)}\ge \frac{F^+(p_{x_1}-x_2)}{F^+(p_{x_2}-x_2)}
\end{equation}
To get the desired contradiction, it suffices to notice that 
$$
\frac{p_{x_2}-v}{p_{x_1}-v}=1+\frac{p_{x_2}-p_{x_1}}{p_{x_1}-v}
$$
is strictly increasing in $v\in\RR$, for $v<p_{x_1}$, and that
$$
\frac{F^+(p_{x_1}-x)}{F^+(p_{x_2}-x)}
$$
is decreasing in $x$. 
The former is obvious, while the latter follows from
\begin{multline*}
\frac{\partial}{\partial x}\left(\frac{F^+(p_{x_1}-x)}{F^+(p_{x_2}-x)}\right)=
\frac{f(p_{x_1}-x)F^+(p_{x_2}-x)-f(p_{x_2}-x)F^+(p_{x_1}-x)}{F^+\left(p_{x_2}-x\right)^2}=\\\frac{f(p_{x_1}-x)f(p_{x_2}-x)}{F^+\left(p_{x_2}-x\right)^2}
\left(
\frac{F^+}{f}(p_{x_2}-x)-\frac{F^+}{f}(p_{x_1}-x)
\right)\le 0,
\end{multline*}
which, in turn, follows from the fact that $F^+/f$ is decreasing, by Assumption \ref{ass.Foverfmono}.
Given the above monotonicity properties of the terms in (\ref{eq.sect4.pZmonoratiogiven}), we deduce (\ref{eq.sect4.pZmonoratiodesired}), thus, obtaining the desired contradiction and proving the monotonicity of $P^a(v)(\cdot)$.

{\bf Step 2.} One can easily check that $P^a(v_1)\ge P^a(v_2)$, if $v_1(x)\ge v_2(x)$ for all $x$. To show that $P^a(v_1)$ and $P^a(v_2)$ are close in the topology of $B^a_0$, it suffices to show that 
$$
\int_0^1 \left\vert P^a(v_1)-P^a(v_2) \right\vert \text{d}x
$$ 
is small.
Note that the monotonicity and the 1-shift property of the integer-valued function $P^a(v_2)$ (cf. Lemma \ref{lemma.responsecontrolprops}) imply that it coincides with $\lfloor x-\alpha_2 \rfloor$, for $x\in[0,1]$ (except, possibly, the jump points of the latter function), for some $\alpha_2$.
Similar conclusion holds for $P^a(v_1)$, with some $\alpha_1$. W.l.o.g. we assume that $\alpha_1\ge\alpha_2$, and, hence, $v_2\ge v_1$.
Assume that $v_1$ and $v_2$ are also $\delta$-close in $\sup$-norm, so that we have $v_2\le v_1+\delta$.
If we can, moreover, show that 
\begin{equation}\label{eq.sect4.alphasdifway}
\alpha_1\le\alpha_2+\frac{\delta}{1-w},
\end{equation}
then, a straightforward calculation would yield 
$$
\int_0^1 \left\vert P^a(v_1)-P^a(v_2) \right\vert \text{d}x = O(\delta).
$$
Thus, it remains to show (\ref{eq.sect4.alphasdifway}). To this end, we note that, due to the definition of $A_0(w)$, for every $x\in\RR$, there exists $x_*\in[x,x+\delta/(1-w)]$, such that $v_1(x_*)=v_2(x)$.
Assuming that 
$$
P^a(v_1)(x_*)<P^a(v_2)(x),
$$ 
and recalling that $x_*\ge x$ and $v_1(x_*)=v_2(x)$, we follow the arguments in \emph{Step 1} to obtain a contradiction.
Thus, 
$$
P^a(v_1)(x_*)\ge P^a(v_2)(x),
$$ 
which implies (\ref{eq.sect4.alphasdifway}).
\qed
\end{proof}

The following lemma, whose proof is given in the appendix, provides the last result we need in order to show the continuity of $\overline{\Phi}$. Recall the definition of $(J^a_0,J^b_0)$ and $(V^a_0,V^b_0)$, given in (\ref{def.Ja0})--(\ref{def.Vb0}).

\begin{lemma}\label{lemma.Vofpcontinuous}
The operators $(p^a,p^b)\mapsto \cdot + V^a_0(\cdot,p^a,p^b,v),\,\cdot + V^b_0(\cdot,p^a,p^b,v)$, from  $B^a_0\times B^b_0$ into $A_0$, are continuous, uniformly over $v\in A_0$.
\end{lemma}

Finally, we can state the main result of this section.

\begin{theorem}\label{thm.fixedpointexistence}
Let $\sigma$ be sufficiently large so that $w$, defined in Proposition \ref{prop.Vmonotonicity}, is strictly less than one.
Then, the set $A:=A_0(w)\times A_0(w)$ is a compact closed convex subset of $C(\mathbb{R})\times C(\mathbb{R})$, with the topology of uniform convergence on all compacts. Moreover, the mapping $\overline{\Phi}$, defined in (\ref{def.overlinePhi}), is a continuous mapping of $A$ into itself, and it has a fixed point.
\end{theorem}

\begin{proof}
The fact that $\overline{\Phi}$ maps $A$ into itself, for sufficiently large $\sigma$, is discussed in the paragraph following Definition \ref{def.A0}.
The closeness and convexity of $A$ are also clear.
The compactness of $A$ follows from the compactness of $A_0(w)$. In turn, $A_0(w)$ is compact because of the uniform Lipschitz property of its elements and their uniform closeness to $x$.
Finally, $\overline{\Phi}$ is continuous because it can be written as a composition of
\begin{equation*}
\mathbf{E}\colon (v^a,v^b)\mapsto \left(v^a,v^b,P^a(v^a),P^b(v^b)\right)
\end{equation*}
and
$$
\mathbf{V}\colon (v^a,v^b,p^a,p^b)\mapsto \left(V^a(\cdot,p^a,p^b,v^b),V^b(\cdot,p^a,p^b,v^a)\right).
$$
In the above, $\mathbf{E}\colon A\to A\times B^a_0\times B^b_0$ is continuous by Lemma \ref{lemma.PofVcontinuous}. The operator
$\mathbf{V}\colon A\times B^a_0\times B^b_0\to A$ is continuous, as it is continuous in $(p^a,p^b)\in B^a_0\times B^b_0$, uniformly over $v^a,v^b\in A_0(w)$, by Lemma \ref{lemma.Vofpcontinuous}, and it is continuous in $v^a, v^b\in A_0(w)$, due to Proposition \ref{prop.VcontinuousinJ} and the 1-shift property of the elements of $A_0(w)$.
The existence of a fixed point of $\overline{\Phi}$ follows from the Schauder fixed point theorem. \qed
\end{proof}

Combining the above theorem with Proposition \ref{prop.Va0=Vatrue}, we obtain the proof of Theorem \ref{thm:main} (compare (\ref{def.overlinePhi}), (\ref{eq.prop.Va0=Vatrue.eq1}), (\ref{eq.prop.Va0=Vatrue.eq2}) and (\ref{eq.2agentjointproblem})).


\section{Numerical example and applications}
\label{se:example}

In this section, we consider a numerical example which illustrates the properties of equilibria constructed in the preceding sections and shows their potential applications. As described in the proof of Theorem \ref{thm.fixedpointexistence}, to compute equilibrium value functions $(\bar{V}^a,\bar{V}^b)$, we need to find a fixed point of a mapping $\mathbf{V}\circ \mathbf{E}$. Notice that $\mathcal{A}^a(x)/\mathcal{A}^b(x)$ are finite sets, which become sufficiently small under realistic assumptions on model parameters $(\sigma,\lambda,F,\alpha)$. Hence, $\mathbf{E}$ can be easily computed by a simple grid search, and we only need to figure out how to compute $\mathbf{V}$. The latter is equivalent to computing a value function of a stationary optimal stopping problem with running costs and discounting. The algorithm we use to compute this value function is described below, and it constitutes a simple application of the Markov chain approximation methods developed in \cite{Kushner}.

Let $N$ be a positive integer which controls the space and time discretization, and define
\begin{equation*}
h:=\frac{1}{N},\quad \Delta t:=\frac{h^2}{\sigma^2}. 
\end{equation*}
Let $\xi_n$ be a symmetric random walk on $\left\{nh\right\}_{n\in\mathbb{Z}}$. It is easy to check that the conditional first and second moments of $\xi_{n+1}-\xi_n$ approximate $X_{(n+1)\Delta t}-X_{n\Delta t}$, and, hence, $\xi$ can be thought of as an approximation of $X$, where each step of $\xi$ is considered to take time $\Delta t$. 
We, then, consider the associated discretization of the optimal stopping problems (\ref{def.VagivenJb}), (\ref{def.VbgivenJa}), for the approximating Markov chain $\xi$, and denote the corresponding value functions by $V^a_N(x,p^a,p^b,v)$ and $V^b_N(x,p^a,p^b,v)$. It is known (cf. \cite{Kushner}) that $V^{a/b}_N(x,p^a,p^b,v)\to V^{a/b}(x,p^a,p^b,v)$, as $N\to\infty$, and that $V^a_N$ satisfies the following dynamic programming equation:
\begin{equation}\label{eq.markchaindynprog}
V^a_N(x,p^a,p^b,v)=\max\left[\lfloor v(x)\rfloor, 
\frac{1}{1+c_p(x)\Delta t}\left(\frac{V^a_N(x-h)+V^a_N(x+h)}{2}\right)+\frac{c_p(x)\Delta t}{1+c_p(x)\Delta t}g^a_p(x),
\right]
\end{equation}
and similarly for $V^b_N$. The solution to this equation can be found via the usual iteration in value space. That is, if we replace $V^a_N$ with the $n$-th step approximation $V^a_{N,n}$, in the right hand side of (\ref{eq.markchaindynprog}), and with $V^a_{N,n+1}$, in its left hand side, then $V^a_{N,n}\to V^a_N$ as $n\to\infty$. 
Note that, to describe an equilibrium (more precisely, its approximation), we need to find $\bar{V}^a_N = V^a_N(\cdot,P^a(V^a_N),P^b(V^b_N),V^b_N)$ and $\bar{V}^b_N = V^b_N(\cdot,P^a(V^a_N),P^b(V^b_N),V^a_N)$, where each $V^{a/b}_N$ solves the associated dynamic programming equation  (\ref{eq.markchaindynprog}). To solve the resulting coupled system, we start with initial $(V^a_{N,0}, V^b_{N,0}, p^a_{N,0}, p^b_{N,0})$, use it to obtain $V^a_{N,1}$, as described above (i.e. by computing the right hand side of (\ref{eq.markchaindynprog}), with $(V^a_{N,0},V^b_{N,0}, p^a_{N,0}, p^b_{N,0})$), compute $p^a_{N,1}=P^a(V^a_{N,1})$ via (\ref{def.responsecontrols}), then, compute $V^b_{N,1}$ given $(V^a_{N,1}, V^b_{N,0}, p^a_{N,1}, p^b_{N,0})$, and so on. If the resulting sequence converges, the limiting function satisfies the definition of $(\bar{V}^a_N,\bar{V}^b_N)$.

Although the theoretical convergence results of \cite{Kushner} do not quite apply for the coupled system at hand, numerically, we do observe convergence of the proposed scheme for the values of $\sigma$ that are not too small, as suggested by Theorem \ref{thm:main}.
The left hand side of Figure \ref{fig:1} shows a typical graph of the (approximated) equilibrium value functions $(\bar{V}^a_N,\bar{V}^b_N)$. Its right hand side contains the associated bid and ask prices $(p^a=P^a(\bar{V}^a_N),p^b=P^b(\bar{V}^b_N))$.
The parameter values are as follows: $f(x)=\bone_{[-\gamma,\gamma]}(x)/(2\gamma)$, $\gamma=1.2$, $\alpha=0.9$, $\lambda=1$, $\sigma=1$, $N=100$. Due to the $1$-shift property of both values and prices, one can easily extend these graph beyond the range $x\in[0,1]$: the value of these functions at $x+n$ is obtained by shifting their values at $x$ up by $n$.
We can observe several properties of the equilibrium value functions and quotes that can be deduced from the construction of equilibrium in the preceding sections.
\begin{itemize}
\item First, the value functions $\bar{V}^{a/b}$ are monotone and Lipschitz in $x$. This follows from the fact that they belong to $A_0(w)$ (cf. Definition \ref{def.A0}), by Theorem \ref{thm.fixedpointexistence}.
\item At the points where the two value functions coincide (and, hence, the game stops), they take integer values. This follows from the arguments presented in the proof of Proposition \ref{prop:equivEquil} and in the paragraph preceding equation (\ref{eq.2agentjointproblem}).
\item The two-dimensional stochastic process $(\bar{V}^a(X_t),\bar{V}^b(X_t))$ evolves on a one-dimensional manifold (as $X$ is one-dimensional), which is shown in red on Figure \ref{fig:3}.
\item The quotes, $p^{a/b}$, take only integer values by construction, and they are non-decreasing as shown in the proof of Lemma \ref{lemma.PofVcontinuous}. In addition, $p^{a/b}(x)$ can change at most once in the interval $x\in(0,1)$, as follows from the $1$-shift property.
\item The ask quote, $p^a(x)$, never falls below $\bar{V}^{a}(x)$, and the bid quote, $p^b(x)$, never rises above $\bar{V}^{b}(x)$, as follows from the definition of the feedback functionals in (\ref{def.responsecontrols}). Since the open interval $(\bar{V}^a(x),\bar{V}^b(x))$ does not contain any integers (see the paragraph preceding (\ref{eq.2agentjointproblem})), we have $p^a(x)\geq p^b(x)$.
\item Finally, if $p^a(x)=\lceil x \rceil$ and $p^b(x)=\lfloor x \rfloor$ (which is expected for markets with a single-tick spread), the value functions always cross: $\bar{V}^a(x)\leq \bar{V}^b(x)$. To see this, notice that, in the equilibrium we construct, the two agents stop at the same time $\tau$, hence, their payoffs are given by 
$$
U^{a/b}_{\tau}\left(p^a(X_{\tau}),\lceil \bar{V}^a(X_{\tau})\rceil,p^b(X_{\tau}),\lfloor\bar{V}^b(X_{\tau})\rfloor\right)
$$ 
(see (\ref{eq.game.main.def.2}) and (\ref{eq.UL.def})). As $ \bar{V}^a(X_{\tau})= \bar{V}^b(X_{\tau})$, the desired inequality follows from $g^a(p^a(x),p^b(x),x)\leq g^b(p^a(x),p^b(x),x)$ (see (\ref{def.gapapbx}), (\ref{def.gbpapbx})), which, in turn, follows from
$$
\int^{\infty}_{\lceil x \rceil-x} \lceil x+\alpha y\rceil\text{d}F(y) \geq \lceil x \rceil \left( 1-F\left(\lceil x \rceil-x\right)\right),
\quad \int_{-\infty}^{\lfloor x \rfloor-x} \lfloor x+\alpha y\rfloor\text{d}F(y) \leq \lfloor x \rfloor F\left(\lfloor x \rfloor-x\right).
$$
\end{itemize}

It is worth mentioning that, in the equilibrium shown in Figure \ref{fig:1}, the value functions coincide, and the game stops, when $X$ hits zero or one. In general, this does not have to be the case, and numerical experiments reveal the existence of equilibria with different stopping boundaries (even though the stopping region is always $1$-periodic). 

Let us now discuss the implications of the last property in the above list, which illustrate the potential applications of the proposed model. As mentioned in the introduction, the fact that the value functions of strategic agents cross, i.e. $\bar{V}^a(X_t)\leq \bar{V}^b(X_t)$, before the game ends is due to the presence of a positive tick size. Without the friction caused by the latter, the two strategic agents would trade with each other at any price level in $[\bar{V}^a(X_t),\,\bar{V}^b(X_t)]$, and the game would end. However, in the present model, the agents may not be able to trade because there may not exist an admissible price level (i.e. an integer) lying between their continuation values. The size of the value crossing, i.e., $\bar{V}^b-\bar{V}^a$, measures the \emph{inefficiency} created by the positive tick size. Namely, if the two players were offered to trade with each other in a ``shadow market", without a tick size, the maximum fee they would be willing to pay for such an opportunity is precisely equal to the size of the crossing. The proposed model allows us to compute the maximum value of such inefficiency, throughout the game, which corresponds to the maximum difference between the red and blue curves in Figure \ref{fig:1}, and to the maximum deviation of the red curve from the diagonal in Figure \ref{fig:3}.
This measure of inefficiency is plotted on Figure \ref{fig:4} as a function of tick size. Remarkably, this function appears to be super-linear, which, in particular, means that the size of the inefficiency is not simply proportional to the tick size.
Notice that any change in the tick size has a dual effect on the market. On the one hand, under the natural assumption that the game always stops when $X$ hits a multiple of the tick size,\footnote{Numerical experiments indicate that this is the most likely equilibrium.} as in Figure \ref{fig:1}, we conclude that the duration of the game is increasing with the tick size. Thus, an increase in the tick size has a positive effect on the market liquidity: the strategic agents withdraw liquidity less frequently. On the other hand, as the inefficiency caused by a positive tick size increases, the strategic agents are more tempted to leave the game altogether and trade at an alternative market with a smaller tick size, even if they pay a fee for such a transition. The equilibrium model developed herein allows one to quantify these two effects, which, for example, can be used by policy makers when making decisions on regulating the tick size in financial exchanges.

\begin{figure}
\begin{center}
  \begin{tabular} {cc}
    {
    \includegraphics[width = 0.45\textwidth]{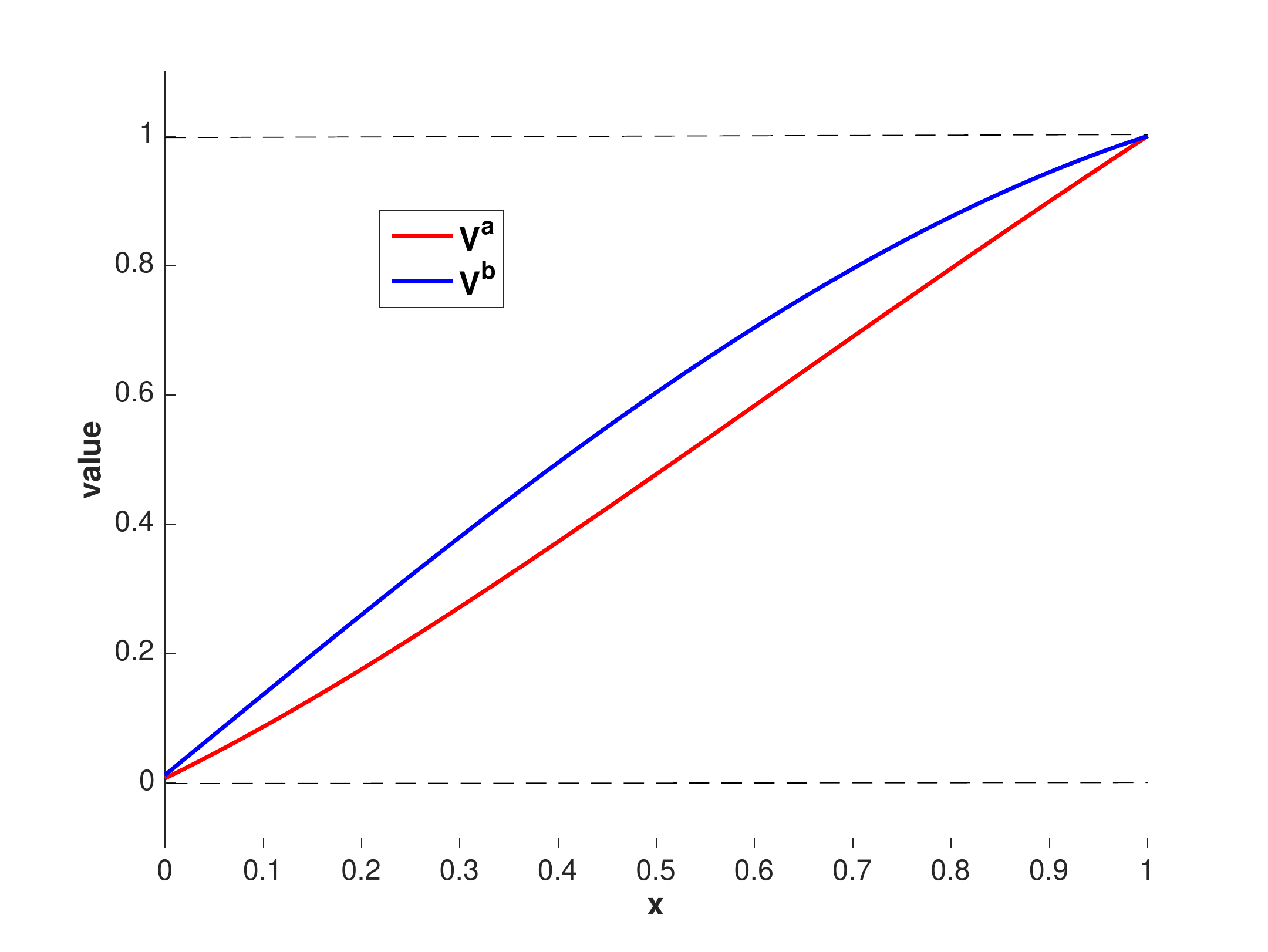}
    } & {
    \includegraphics[width = 0.45\textwidth]{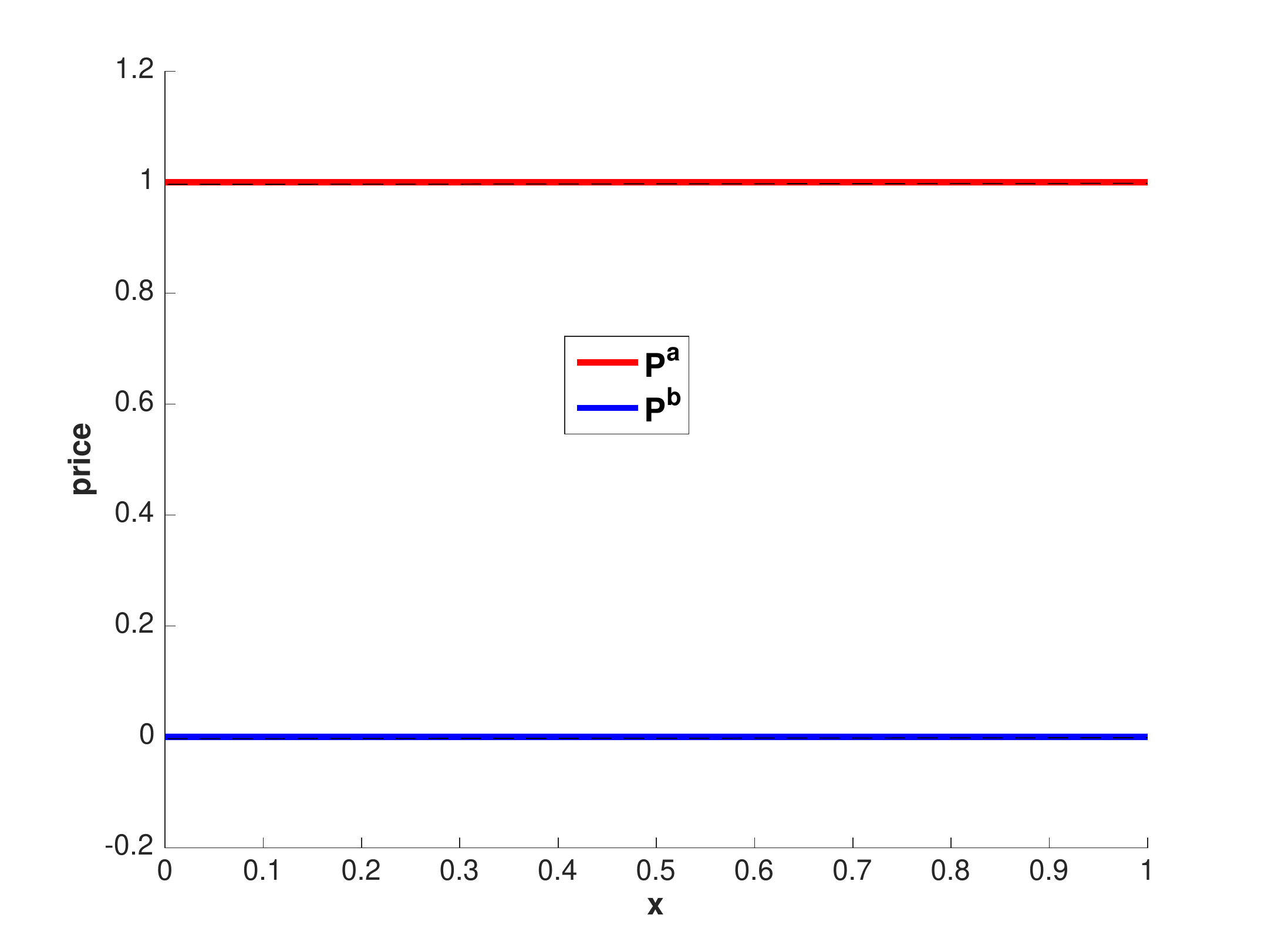}
    }\\
  \end{tabular}
  \caption{On the left: value functions $(\bar{V}^b_N,\bar{V}^a_N)$ as functions of $X$. On the right: posted prices $(p^a,p^b)$, as functions of $X$.}
    \label{fig:1}
\end{center}
\end{figure}


\begin{figure}
\begin{center}
  \begin{tabular} {cc}
    {
    \includegraphics[width = 0.45\textwidth]{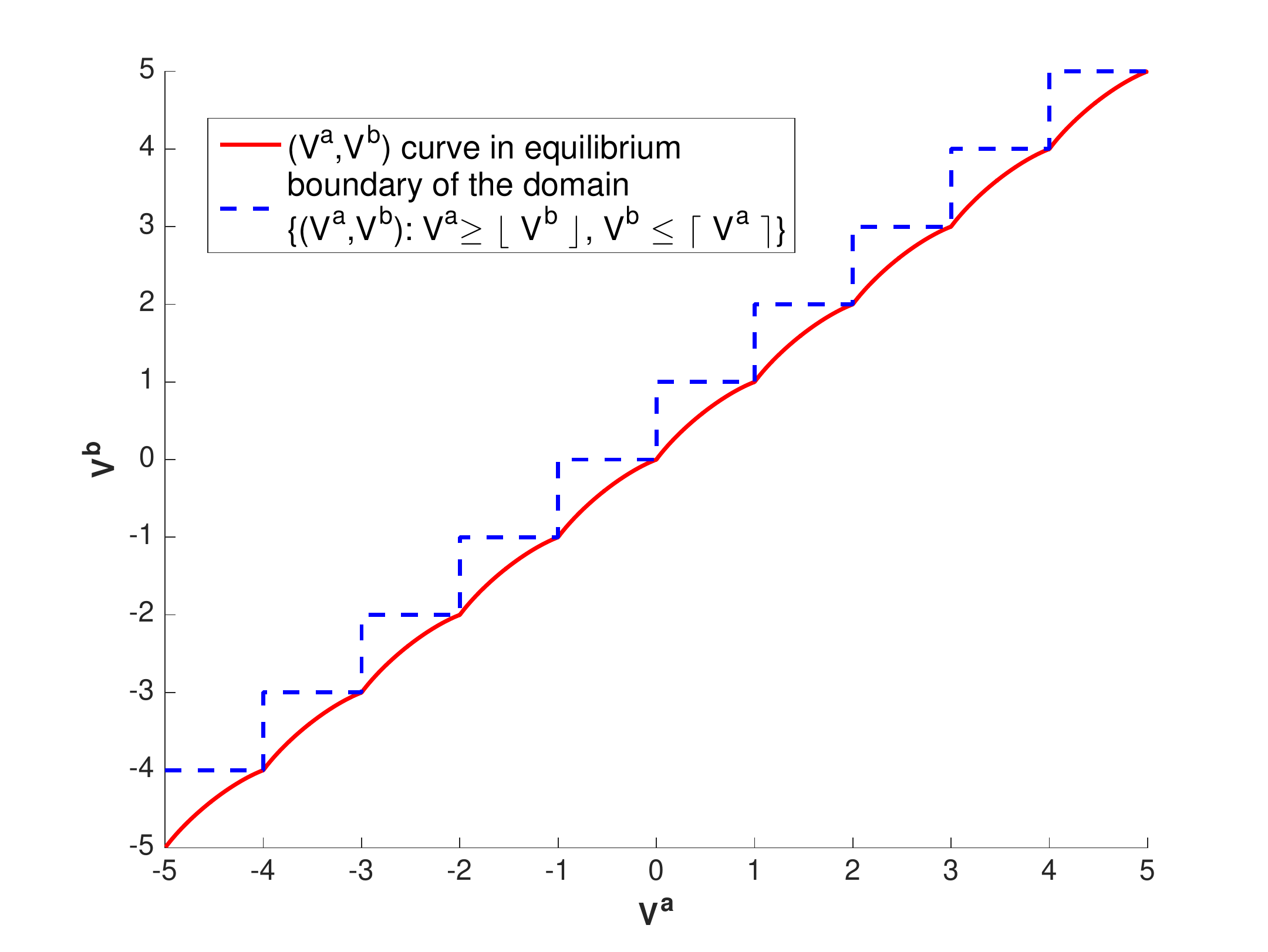}
    } & {
    \includegraphics[width = 0.45\textwidth]{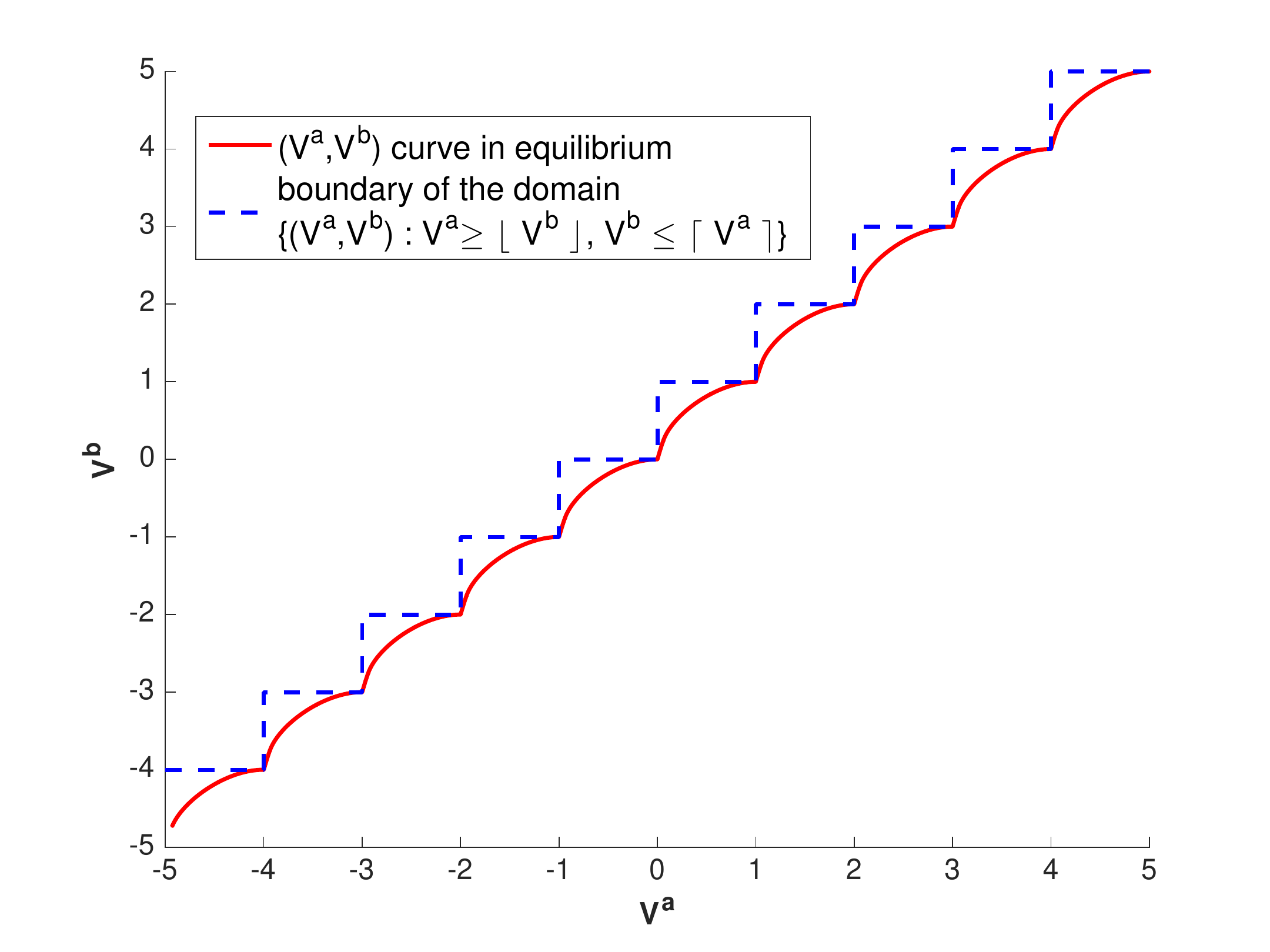}
    }\\
  \end{tabular}
  \caption{In red: the curve $\left\{(\bar{V}^a(x),\bar{V}^b(x))\,:\,x\in\RR\right\}$. In blue: the boundary of the set $\{(x,y)\,:\,x\geq \lfloor y \rfloor,\,y\leq \lceil y \rceil\}$. The left graph uses the same parameter values as in Figure {fig:1}.}
    \label{fig:3}
\end{center}
\end{figure}

\begin{figure}
\begin{center}
    \includegraphics[width = 0.45\textwidth]{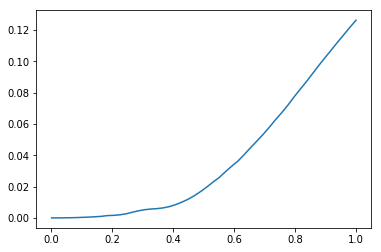}
  \caption{The size of inefficiency caused by a positive tick size, as a function of tick size.}
    \label{fig:4}
\end{center}
\end{figure}

\section{Appendix.}

\subsection{Proof of Proposition \ref{prop.Va0=Vatrue}}

Subtracting $x$ from $\overline{V}^a$ and $J^a$, as in (\ref{def.Va0}) and (\ref{def.Ja0}), we reformulate the claim for the relative versions, $\overline{V}^a_0=\overline{V}^a - x$ and $\overline{J}^a_0=J^a-x$ in lieu of $\overline{V}^a$ and $J^a$.
To prove the claim, we need to verify that the feedback control $p^a=P^a(\overline{V}^a)$ is optimal. The latter, in turn, requires a differential characterization of the value function 
$$
V^a_0(x,p^a,p^b,\lfloor v\rfloor) = \sup_{\tau} J^a_0(x,\tau,p^a,p^b,\lfloor v\rfloor),
$$ 
for any admissible control $p^a$ (cf. (\ref{def.Va0}), (\ref{def.Ja0})), along with a comparison principle allowing us to compare the candidate optimal control to an arbitrary one. We use the theory of variational inequalities (VIs) to implement this program. Unfortunately, we could not locate any VI results that would be directly applicable in our setting, due to the presence of unbounded domain, $\mathbb{L}^\infty$ discount factor and running costs, and discontinuous obstacle. Namely, to the best of our knowledge, there exist no results that would connect the solution to an optimal stopping problem and the VI solution in such a setting. Thus, we need to introduce additional approximation steps.
More specifically, in \emph{Step 1}, we show that the discontinuous barrier $\lfloor v \rfloor$ can be replaced by its continuous majorating approximation $s_{\epsilon}$, without affecting the value of the associated optimal stopping problem, no matter which admissible $(p^a,p^b)$ are chosen. In \emph{Step 2}, we choose a sequence of smooth approximating functions $v_n\downarrow s_\epsilon$ and use the continuity of VI solution w.r.t. the obstacle (cf. \cite{Bens}), to show that the value function corresponding to $s_\epsilon$ satisfies an appropriate VI. Finally, in \emph{Step 3} we use the comparison results for the associated VI, to show that $P^a(\overline{V}^a)$ is indeed the optimal control.

\vskip 3pt

{\bf Step 1.} This step is taken care of by the following lemma, whose (geometric) proof is given below.
\begin{lemma}\label{lemma.vbfloortosepsilon}
Let $(p^a,p^b)$ be admissible controls, and $v$ an admissible barrier, such that $(p^a,p^b,v)$ have 1-shift property, $g^{a/b}/c$ are $C$-close to $x$, and $v'\ge 1-w$, with some fixed constant $w<1$, and such that Proposition \ref{prop.fabderbounds} holds with such $w$ (the latter can be guaranteed by choosing a sufficiently large $\sigma$, due to Proposition \ref{prop.fabderbounds}). Then, there exists a continuous piecewise linear function $s_\epsilon\ge \lfloor v \rfloor$, independent of $(p^a,p^b)$, which satisfies the 1-shift property, is $C$-close to $x$, and is such that
$$
V^a_0(\cdot,p^a,p^b,\lfloor v \rfloor)=V^a_0(\cdot,p^a,p^b,s_{\epsilon}),
$$
for any $(p^a,p^b)$ satisfying the above properties.
\end{lemma}
\begin{proof}

The main idea of the proof is to modify a step function around its jump points, by replacing jumps with steep line segments, so that we do not affect its `$\operatorname{mcm}$'.
For notational convenience, let us introduce
\begin{multline}\label{def.fa0}
f^a_0(x,p^a,p^b):=f^a(x,p^a,p^b)-x=\\
\EE^x\left[
\int_0^\infty \exp\left(-\int_0^t c(X_s,p^a(X_s),p^b(X_s))\text{d}s\right)
\left(g^a(p^a(X_t),p^b(X_t),X_t)-c(p^a(X_t),p^b(X_t),X_t)X_t
\right)
\text{d}t
\right].
\end{multline}
Then, by Proposition \ref{prop.vabmcmcharacterization}, we have:
$$
V^a_0(x,p^a,p^b,\tilde{v})=f^a_0(x)+\phi(x)\operatorname{mcm}\left(\widehat{\tilde{v}-f^a_0-x}\right)\left(\Fr(x)\right),
$$
where the only dependence on the obstacle $\tilde{v}$ is inside the `$\operatorname{mcm}$', and $V^a_0$ is defined in (\ref{def.Va0}).
Thus, it suffices to show that $\operatorname{mcm}\left(\widehat{\tilde{v}-f^a_0-x}\right)\left(y\right)$ does not change if we replace $\tilde{v}=\lfloor v \rfloor$ by $\tilde{v}=s_{\epsilon}$.
	
Let us fix any $\epsilon>0$ and define $s_{\epsilon}$. We know that $\lfloor v \rfloor$ has the $1$-shift property and changes only by jumps of size $1$, at a sequence of points $\left\{x_0+n\right\}_{n\in\mathbb{Z}}$.
We define $s_{\epsilon}(x)$ to be equal to $\lfloor v(x) \rfloor$ outside the intervals $(x_0+n-\epsilon,x_0+n]$, and to coincide with the line segment connecting $(x_0+n-\epsilon,\lfloor v\left(x_0+n-\epsilon\right)\rfloor)$ and $(x_0+n,\lfloor v\left(x_0+n\right)\rfloor)$ on these intervals.
Notice that, in the left $\epsilon$-neighborhood of every jump point of $\lfloor v \rfloor$, $s_\epsilon$ is a line segment with slope $1/\epsilon$, and it coincides with $\lfloor v \rfloor$ (and is locally constant) elsewhere.
Notice also that $s_\epsilon\ge \lfloor v \rfloor$ by construction, hence, $\operatorname{mcm}\left(\widehat{s_\epsilon-f^a_0-x}\right)\ge \operatorname{mcm}\left(\widehat{\lfloor v \rfloor-f^a_0-x}\right)$, and it only remains to prove the opposite inequality, for some $\epsilon>0$.

To this end, we notice that, under the assumptions of the lemma, $f^a_0+x$ is strictly increasing, and, hence, the function $x\mapsto\lfloor v(x) \rfloor-f^a_0(x)-x$ achieves its maximum value exactly at the points $\left\{x_0+n\right\}$. If this maximum is non-positive, we deduce from the proof of Proposition \ref{prop.Vmonotonicity} that $\operatorname{mcm}\left(\widehat{\lfloor v \rfloor-f^a_0-x}\right)=0$. Then, the desired inequality is clear, as we can ensure that $s_\epsilon-f^a_0-x$ has the same supremum as $\lfloor v \rfloor-f^a_0-x$, by choosing sufficiently small $\epsilon>0$ (depending only on $w$, due to Proposition \ref{prop.fabderbounds}). Thus, for the remainder of the proof, we assume that the supremum of $\lfloor v \rfloor-f^a_0-x$ is strictly positive, which implies that its `$\operatorname{mcm}$" is strictly positive everywhere.

Denote $y_0:=\mathbf{F}(x_0)$, $y_1:=\mathbf{F}(x_1)$, where $x_1=x_0+1$ and $\mathbf{F}$ is given by (\ref{def.F}). Notice that $(f^a_0)'\geq 1-w$, with some fixed $w<1$. Then, using the estimates of $(\psi,\phi)$ given in Lemma \ref{lemma.phipsiasymptotics}, it is easy to deduce that there exists $\delta\in(0,1]$, independent of $(p^a,p^b,v)$, and $y_0'\in[y_0,\mathbf{F}(x_1-\delta)]$, such that $\widehat{(\lfloor v \rfloor-f^a_0-x)}(y_0')>0$ and the linear interpolation of the points $(y_0',\widehat{(\lfloor v \rfloor-f^a_0-x)}(y_0'))$ and $(y_1,\widehat{(\lfloor v \rfloor-f^a_0-x)}(y_1))$ lies above the graph of $\widehat{\lfloor v \rfloor-f^a_0-x}$ on the interval $[y_0',y_1]$.
Our goal is to show that this interpolation also dominates $\widehat{s_\epsilon-f^a_0-x}$ on this interval, for all small enough $\epsilon>0$. Clearly, if $\epsilon<\delta$, then, the function $\widehat{\lfloor v \rfloor-f^a_0-x}$, modified by the aforementioned interpolation on $[y_0',\mathbf{F}(x_0+1)]$, dominates $\widehat{s_\epsilon-f^a_0-x}$ on $[\mathbf{F}(x_0),\mathbf{F}(x_0+1)]$. Then, due to periodicity, analogous modification can be constructed on any $[\mathbf{F}(x_0+n),\mathbf{F}(x_0+n+1)]$, to dominate $\widehat{s_\epsilon-f^a_0-x}$ on this interval, which yields the desired inequality: $\operatorname{mcm}\left(\widehat{s_\epsilon-f^a_0-x}\right)\le \operatorname{mcm}\left(\widehat{\lfloor v \rfloor-f^a_0-x}\right)$.

Notice that the `$\,\widehat{\cdot}\,$' operator (defined in (\ref{def.widehat})) transforms any linear combination of $\phi$ and $\psi$ into a straight line.
Therefore, it suffices to prove that $s_\epsilon-f^a_0-x$ is dominated on $[x_0':=\mathbf{F}^{-1}(y_0'),x_1]$ by `$a\phi(x)+b\psi(x)$' interpolation between $(x_0',z_0')$ and $(x_1,z_1)$ (with the constants $a$ and $b$ chosen to match the boundary values $z_0'$ and $z_1$, at $x_0'$ and $x_1$, respectively), where 
$$
z_0' = \lfloor v(x_0') \rfloor - f^a_0(x_0') - x_0', \quad
z_1=\lfloor v(x_1) \rfloor - f^a_0(x_1) - x_1.
$$ 
Denote $h = a\phi + b\psi$, and observe: $h(x_0')=z_0'\in[(z_1-1)^+,z_1]$, $h(x_1)=z_1>0$, $x_1\geq x_0'+\delta$, and $h\geq0$, $h\geq \lfloor v \rfloor-f^a_0-x$ on $[x_0',x_1]$.
Next, as $h$ satisfies
$$
\frac{\sigma^2}{2}h_{xx}-ch=0,\,\,\text{a.e.},
$$
and is continuously differentiable, we apply the maximum principle to conclude that the maximum value of $h$ on $[x_0',x_1]$ cannot be achieved in the interior of this interval (otherwise, at the maximum point, we would have $h>0$ and $h_{xx}<0$, which is a contradiction with the above equation).
Thus, $h(x)\le z_1$ for $x\in[x_0',x_1]$. Then, the above equation implies
$$
0\leq h_{xx}(x)\le \frac{2c}{\sigma^2}z_1\le \frac{2c_u}{\sigma^2}z_1,\quad x\in[x_0',x_1],
$$
and, as the integral of $h_x$ over the interval $[x_0',x_1]$ (of length at least $\delta\in(0,1]$ and at most $1$) does not exceed $1$, we conclude:
$$
h_x(x)\le \frac{2c_u}{\sigma^2}z_1 + \frac{1}{\delta},\quad x\in[x_0',x_1].
$$
The assumptions of the lemma imply that $\lfloor v \rfloor$ is $(C+1)$-close to $x$ and that $f^a_0$ is absolutely bounded by $C$. Therefore, $z_1= \lfloor v(x_1) \rfloor - x_1 - f^a_0(x_1)\le 2C+1$, and, in turn, $h_x$ is bounded on $[x_0',x_1]$ by a constant $C_2$, independent of $(p^a,p^b,v)$, as long as the latter satisfy the properties stated in the lemma.
Finally, notice that $h(x_1)=z_1=s_\epsilon(x_1)-x_1-f^a_0(x_1)$, and that the slope of the function $s_\epsilon-f^a_0-x$ is at least as large as $1/\epsilon-1-w$ on $[x_1-\epsilon,x_1]$ (due to Proposition \ref{prop.fabderbounds}). Then, we can choose $\epsilon$ to be small enough, so that $1/\epsilon-1-w\geq C_2\geq h_x(x)$, for $x\in[x_1-\epsilon,x_1]$.
Thus, we obtain
$$
s_\epsilon(x)-f^a_0(x)-x\le h(x),\,\,x\in[x_1-\epsilon,x_1].
$$
As $h\geq \lfloor v \rfloor-f^a_0-x$ on $[x_0',x_1]$, by construction, we conclude that the above inequality holds for all $x\in[x_0',x_1]$.
Thus, we conclude that the function $\widehat{\lfloor v \rfloor-f^a_0-x}$, modified by its linear interpolation on $[y_0',y_1]$, dominates $\widehat{s_\epsilon-f^a_0-x}$ on $[y_0=\mathbf{F}(x_0),y_1=\mathbf{F}(x_0+1)]$.
The other intervals $[\mathbf{F}(x_0+n),\mathbf{F}(x_0+n+1)]$ are handled by the same argument.
As a result, we obtain
$$
\operatorname{mcm}\left(\widehat{s_\epsilon-f^a_0-x}\right)\le 
\operatorname{mcm}\left(\widehat{\lfloor v \rfloor-f^a_0-x}\right),
$$
which yields the desired equality between the left and the right hand sides of the above.
\qed
\end{proof}

{\bf Step 2.} Let us recall some notation from \cite{Bens}. Let $\mu>0$ and consider the weight function
$$
m_\mu(x)=\exp(-\mu |x|).
$$ 
Denote by $H_\mu=\mathbb{W}^{0,2,\mu}$, $V_\mu=\mathbb{W}^{1,2,\mu}$ the appropriate $m_\mu$-weighted Sobolev spaces on $\mathbb{R}$ (we need weighted spaces because our coefficients are bounded and periodic, while we need them to be integrable over the entire unbounded domain).
For any $u,r\in V_\mu$ and any $p\in\mathcal{A}^a$, we define
$$
a_p(u,r)=\int_{\RR} \frac{\sigma^2}{2}m^2_\mu u' r' 
- 2\mu\operatorname{sign}(x)\frac{\sigma^2}{2}m_\mu v' m_\mu r 
+ c_p m^2_\mu u r\,dx.
$$
Let $f_p\in H_\mu$, given by
$$
f_p(x)=g^a_p(x)-c_p(x)x,
$$
be the running cots of our relative-to-$x$ stopping problem, and
$$
K_\mu(v)=\left\{u\in V_\mu | u(x)\ge v(x)-x \text{ a.e. }x\right\}
$$
be the appropriate set of test functions.

We denote by $\text{VI}(p,v)$ the following VI (understood in the weak sense)
\begin{equation}\label{eq.VIpJbrel}
a_p(u,r-u)\le (f_p,r-u),\,\,\forall r\in K_\mu(v)
\end{equation}
where 
$$
(u,r)=\left(u,r\right)_\mu=\int_{\RR} m^2_\mu ur dx.
$$
We say that $u$ is a solution of the above VI if $u\in K_\mu(v)$ and $u$ satisfies (\ref{eq.VIpJbrel}).
As all the coefficients, $c_p$, $\sigma^2/2$, $f_p$, and $(v-x)$, are in $\mathbb{L}^\infty(\mathbb{R})$, and as the form $a_p(\cdot,\cdot)$ is coercive when $\mu$ is sufficiently small, we conclude that (for such $\mu$) the VI (\ref{eq.VIpJbrel}) has a unique solution in $K_\mu(v)$, for any $p\in\mathcal{A}^a$ and any $v$ that is $C$-close to $x$, by Theorem 1.13, \cite{Bens} p. 217. 

Let $v_n$ be a $C^\infty$-approximation from above of $s_\epsilon$, associated with $v$ by Lemma \ref{lemma.vbfloortosepsilon}, which is at most $1/n$ away from $s_\epsilon$ in $\sup$-norm. Then, by Theorem 3.19, \cite{Bens} p. 387, for sufficiently small $\mu$, $u_n=V^a_0(\cdot,p,p^b,v_n)$ is the unique solution of $\text{VI}(p,v_n)$. Denote also by $u_0$ the unique solution of $\text{VI}(p,s_\epsilon)$. Rewriting these VIs as unweighted VIs for $m_\mu u$ and restricting to a bounded domain, one can generalize Theorem 1.10, \cite{Bens} p. 207, to obtain $u_n\to u_0$ in $\mathbb{L}^\infty(\mathbb{R})$. The latter fact, together with the easy to check convergence of value functions,
$$
V^a_0(\cdot,p,p^b,v_n)\to V^a_0(\cdot,p,p^b,s_\epsilon)=V^a_0(\cdot,p,p^b,\lfloor v\rfloor),
$$ 
implies that the latter value function is the unique solution of $\text{VI}(p,s_\epsilon)$.

{\bf Step 3.} By Theorem 1.4, \cite{Bens} p. 198, extended to unbounded domain as in Remark 1.21, p. 219, the unique solutions $u,\tilde{u}\in K_\mu(v)$ of VIs
$$
a_p(u,r-u)\le (h,r-u),\,\,\forall r\in K_\mu(v)
$$
$$
\text{resp. } a_p(\tilde{u},r-\tilde{u})\le (\tilde{h},r-\tilde{u}),\,\,\forall r\in K_\mu(v)
$$
sharing the obstacle $v$ and the form $a_p$, but with different right-hand sides $h$ and $\tilde{h}$, satisfy $\tilde{u}\ge u$ if $\tilde{h}\ge h$.
Recall that $\overline{V}^a_0(x)=\overline{V}^a(x)-x$ and denote $\bar{p}:=P^a(\overline{V}^a)$.
We need to show that the following inequality holds:
\begin{multline*}
\bar{u}:=\overline{V}^a_0=V^a_0(\cdot,\bar{p},p^b,\lfloor v\rfloor)=
V^a_0(\cdot,\bar{p},p^b,s_\epsilon)
\ge V^a_0(\cdot,p_0,p^b,s_\epsilon)=
V^a_0(\cdot,p_0,p^b,\lfloor v\rfloor)=:\bar{u}_0=\bar{u}_0(p_0)
\end{multline*}
for any $p_0\in\mathcal{A}^a$. It is shown in Step 2 that $\bar{u}$ satisfies a version of (\ref{eq.VIpJbrel}) with the running costs $f_{\bar{p}}$ and the quadratic form $a_{\bar{p}}$, which, after routine transformations, turns out to be equivalent to
\begin{equation}
a_{p_0}(\bar{u},r-\bar{u})\le (\tilde{f},r-\bar{u}),\,\,\forall r\in K_\mu(\lfloor v \rfloor),
\end{equation}
where $\tilde{f}=f_{p_0}+q$,
$$
q:=g_{\bar{p}}-c_{\bar{p}}(\bar{u}+x)-(g_{p_0}-c_{p_0}(\bar{u}+x))\ge0,
$$
which follows from $\bar{p}=P^a(\overline{V}^a)$ and $\bar{u}(x) + x = \overline{V}^a(x)$. Recall that $\bar{u}_0$ satisfies the above equation with the running cost function $f_{p_0}$ in place of $\tilde{f}$, and that $f_{p_0}\le\tilde{f}$. Hence, we can apply the comparison principle, stated at the beginning of this step, to complete the proof.
\qed

\subsection{Proof of Lemma \ref{lemma.Vofpcontinuous}}

We only show the continuity of $(p^a,p^b)\mapsto x + V^a_0(x,p^a,p^b,v)$, the other part being analogous. Recall that 
$$
V^a_0(x,p^a,p^b,v)=\sup_{\tau}J^a_0(x,\tau,p^a,p^b,v).
$$ 
Thus, it suffices to show that $J^a_0$-s corresponding to two close pairs of prices, $(p^a_1,p^b_1)$ and $(p^a_2,p^b_2)$, with the same $\tau$, are also close, uniformly in $\tau$. To this end, we recall (\ref{def.Ja0}) and obtain
\begin{multline}\label{eq.sect4.threepartobjdiffpabcont}
J^a_0(x,\tau,p^a_1,p^b_1,v)-J^a_0(x,\tau,p^a_2,p^b_2,v)\\
=\EE^x\Big[\int_0^\tau \left(\exp\left(-\int_0^t c_1(X_s)\text{d}s\right)-\exp\left(-\int_0^t c_2(X_s)\text{d}s\right)\right)\left(g^a_1(X_t)-c_1(X_t)X_t\right)\text{d}t\\
+\int_0^\tau \exp\left(-\int_0^t c_2(X_s)\text{d}s\right)\left(g^a_1(X_t)-c_1(X_t)X_t-(g^a_2(X_t)-c_2(X_t)X_t)\right)\text{d}t\\
+\left(\exp\left(-\int_0^\tau c_1(X_s)\text{d}s\right)-\exp\left(-\int_0^\tau c_2(X_s)\text{d}s\right)\right)\left(v(X_\tau)-X_\tau\right)
\Big],
\end{multline}
where we denote
$$
c_1(x)=c(p^a_1(x),p^b_1(x),x),\,\,c_2(x)=c(p^a_2(x),p^b_2(x),x),
$$
$$
g^a_1(x)=g^a(p^a_1(x),p^b_1(x),x),\,\,g^a_2(x)=g^a(p^a_2(x),p^b_2(x),x).
$$
To complete the proof, it suffices to show that the expectations of the absolute values of each of the three terms on different lines in (\ref{eq.sect4.threepartobjdiffpabcont}) are small when $(p^{a}_1,p^b_1)$ and $(p^{a}_2,p^b_2)$ are close in their topology.

For the third term, note that, as $|\exp(-x)-\exp(-y)|\le \max\left(\exp(-x),\exp(-y)\right)|x-y|$, and as $c_i(x)\ge c_l>0$, for all $x$ and $i=1,2$, we have
\begin{multline*}
\left\vert \left(\exp\left(-\int_0^\tau c_1(X_s)\text{d}s\right)-\exp\left(-\int_0^\tau c_2(X_s)\text{d}s\right)\right)\left(v(X_\tau)-X_\tau\right) \right\vert\le\\
\exp(-c_l\tau)\left(\int_0^\tau \left\vert c_1(X_s)-c_2(X_s)\right\vert \text{d}s\right) \left\vert v(X_\tau)-X_\tau \right\vert\le
C_1 \int_0^\tau \exp(-c_l s)\left\vert c_1(X_s)-c_2(X_s)\right\vert \text{d}s\le\\
C_1 \int_0^\tau \exp(-c_l s)\left(\left\vert p^a_1(X_s)-p^a_2(X_s)\right\vert+\left\vert p^b_1(X_s)-p^b_2(X_s)\right\vert\right)\text{d}s,
\end{multline*}
where a positive constant $C_1$ may differ between the lines (here and throughout the proof). The second inequality in the above follows from the closeness to $x$ of $v$, and the last one follows from the fact that
$$
c(p^a,p^b,x)=\lambda\left( F^+(p^a-x)+F(p^b-x)\right)
$$
is Lipschitz in $(p^a,p^b)$, as the density of $\xi$ is bounded by Assumption \ref{ass.Foverfmono}.

Similarly, it is easy to show that, for all $x\in\RR$,
$$
\left\vert (g^a_1(x)-c_1(x)x-(g^a_2(x)-c_2(x)x)\right\vert\le
C_1\left(
\left\vert p^a_1(x)-p^a_2(x) \right\vert
+
\left\vert p^b_1(x)-p^b_2(x) \right\vert
\right),
$$
using the boundedness of the density of $\xi$ and the uniform closeness of all admissible $p^a$, $p^b$ to $x$.
This allows us to estimate the second term in (\ref{eq.sect4.threepartobjdiffpabcont}):
\begin{multline*}
\left\vert
\int_0^\tau \exp\left(-\int_0^t c_2(X_s)\text{d}s\right)\left(g^a_1(X_t)-c_1(X_t)X_t-(g^a_2(X_t)-c_2(X_t)X_t)\right)\text{d}t
\right\vert\le\\
C_1\int_0^\tau \exp(-c_l t) \left(
\left\vert p^a_1(X_t)-p^a_2(X_t) \right\vert
+
\left\vert p^b_1(X_t)-p^b_2(X_t) \right\vert
\right) \text{d}t.
\end{multline*}

Finally, we notice that $\left\vert g^a_1(X_t)-c_1(X_t)X_t \right\vert\le C_1$, which follows from the fact that $g^a_1/c_1$ is $C'_0$-close to $x$ (cf. Lemma \ref{lemma.responsecontrolprops}) and that $c_1\le c_u$. This allows us to estimate the first term in (\ref{eq.sect4.threepartobjdiffpabcont}):
\begin{multline*}
\left \vert
\int_0^\tau\left(\exp\left(-\int_0^t c_1(X_s)\text{d}s\right)-\exp\left(-\int_0^t c_2(X_s)\text{d}s\right)\right)\left(g^a_1(X_t)-c_1(X_t)X_t\right)\text{d}t
\right\vert\le\\
C_1 \int_0^\tau \exp(-c_l t)\left( \int_0^t\left(
\left\vert p^a_1(X_s)-p^a_2(X_s) \right\vert
+
\left\vert p^b_1(X_s)-p^b_2(X_s) \right\vert
\right) \text{d}s\right)\text{d}t\le\\
C_1\int_0^\tau \exp(-c_l t) \left(
\left\vert p^a_1(X_t)-p^a_2(X_t) \right\vert
+
\left\vert p^b_1(X_t)-p^b_2(X_t) \right\vert
\right)\text{d}t,
\end{multline*}
where the second inequality follows from integration by parts, after discarding some negative terms.

Thus, the absolute values of all terms in (\ref{eq.sect4.threepartobjdiffpabcont}) are estimated from above via
\begin{multline*}
\int_0^\tau \exp(-c_l t) \left(
\left\vert p^a_1(X_t)-p^a_2(X_t) \right\vert
+
\left\vert p^b_1(X_t)-p^b_2(X_t) \right\vert
\right) \text{d}t\\
\leq \int_0^\infty \exp(-c_l t) \left(
\left\vert p^a_1(X_t)-p^a_2(X_t) \right\vert
+
\left\vert p^b_1(X_t)-p^b_2(X_t) \right\vert
\right) \text{d}t,
\end{multline*}
which implies
\begin{multline*}
\left\vert V^a_0(x,p^a_1,p^b_1,v)-V^a_0(x,p^a_2,p^b_2,v) \right \vert\\
\leq C_1 \EE^x\left[ \int_0^\infty \exp(-c_l t) \left(
\left\vert p^a_1(X_t)-p^a_2(X_t) \right\vert
+
\left\vert p^b_1(X_t)-p^b_2(X_t) \right\vert
\right) \text{d}t\right].
\end{multline*}
It only remains to estimate the latter expectation in terms of $\mathbb{L}^1\left([0,1]\right)$ norms of $p^a_1-p^a_2$ and $p^b_1-p^b_2$. The latter is achieved easily by interchanging the expectation and integration and using the standard estimates of a Gaussian kernel.
\qed

\bibliographystyle{abbrv}
\bibliography{MFGLOB_impact_refs}

\end{document}